\documentclass[12pt]{article}
\usepackage{graphicx}
\usepackage{float}
\usepackage[flushleft]{caption2}
\usepackage{amsmath,amsfonts,amssymb}
\usepackage{amsthm}
\usepackage{mathrsfs}
\usepackage{enumitem}
\usepackage{tikz}
\usepackage[title]{appendix}
\usepackage{indentfirst}
\usepackage{epstopdf}
\usepackage{hyperref}
\hypersetup{hidelinks}

\DeclareMathOperator*{\Res}{Res}
\usepackage[sort&compress,numbers]{natbib}
\usetikzlibrary{arrows}
\newtheorem{lemma}{Lemma}[section]
\newtheorem{theorem}{Theorem}[section]

\newtheorem{proposition}{Proposition}[section]

\newtheorem{assumption}{Assumption}[section]

\setitemize[1]{itemsep=0pt,partopsep=0pt,parsep=\parskip,topsep=5pt}
\numberwithin{equation}{section}
\newtheorem{RHP}{Riemann--Hilbert problem}
\newcommand{\R}{\mathbb{R}}
\newcommand{\Z}{\mathbb{Z}}
\newcommand{\C}{\mathbb{C}}

\def\d{{\rm d}}

\def\i{{\rm i}}

\def\le{\left}
\def\ri{\right}

\usepackage{tikz}
\usetikzlibrary{arrows.meta,decorations.markings}
\usepackage{tikz}
\usetikzlibrary{shapes,snakes,calc,fit}
\usetikzlibrary{decorations.text}
\usepgfmodule{shapes}
\usetikzlibrary{decorations.pathmorphing}
\usetikzlibrary{decorations.markings}
\usetikzlibrary{patterns}
\usetikzlibrary{automata}
\usetikzlibrary{positioning}
\usetikzlibrary{spy}
\usetikzlibrary{backgrounds}
\tikzset{partial ellipse/.style args={#1:#2:#3}{insert path={+ (#1:#3) arc (#1:#2:#3)} }}
\tikzset{->-/.style={decoration={ markings, mark=at position #1 with {\arrow{>}}},postaction={decorate}}}
\tikzset{-<-/.style={decoration={ markings, mark=at position #1 with {\arrow{<}}},postaction={decorate}}}

\title{\LARGE\bf Large-time asymptotics of a new KdV soliton gas}
\author{\hspace{0.6 cm}{Dedi Yan$^{a}$, Xianguo Geng$^{a,b}$, Jiao Wei$^{a}$\footnote{\footnotesize
			Corresponding author. {\sl Email address}: math.jwei@zzu.edu.cn}}\\
	\leftline{\hspace{0.6 cm}{\small{\sl $^{a}$ School of Mathematics and Statistics, Zhengzhou University, 100 Kexue Road, Zhengzhou, }}}\\
	\leftline{\hspace{0.6 cm}{\small{\sl \quad Henan 450001, People's Republic of China}}}\\
	\leftline{\hspace{0.6 cm}{\small{\sl $^{b}$ Institute of Mathematics, Henan Academy of Sciences, Zhengzhou, Henan 450046,}}}\\
	\leftline{\hspace{0.6 cm}{\small{\sl \quad People's Republic of China}}}}
\date{}
\begin{document}
	\maketitle
	\begin{abstract}
We study the large-time asymptotic behavior of a new KdV soliton gas. We first introduce a pure-soliton Riemann--Hilbert(RH) problem with \(2N\) poles and two different types of residue conditions. We show that, as \(N\to\infty\), this discrete problem converges to  primitive-potential RH problem introduced by Dyachenko, Zakharov, and Zakharov, and the jump matrix of this soliton gas RH problem has two nonzero reflection coefficients.
To analyze the large-time behavior, we apply the Deift--Zhou nonlinear steepest descent method together with an appropriate \(g\)-function mechanism. Through a sequence of transformations, the original RH problem is reduced to explicitly solvable model problems on an associated hyperelliptic Riemann surface. This allows us to derive an explicit leading-order asymptotic formula for the solution in terms of Jacobi elliptic function. The result provides a rigorous asymptotic description of  a new KdV soliton gas and extends the available analysis beyond the previously studied case \(r_2\equiv 0\).\\

	\end{abstract}

\section{Introduction}
The concept of soliton gas, representing the infinite statistical ensemble of interacting solitons, was initially introduced by Zakharov \cite{Zak71}. Building on this foundation, El et al. extended Zakharov's model to describe a dense Korteweg–de Vries (KdV)  soliton gas and breather gases for the focusing nonlinear Schr\"{o}dinger equation by employing spectral theory to analyze the thermodynamic limit of finite-gap solutions \cite{E1,E2,E3,E4,TW2022}. In recent years, soliton gases have attracted considerable attention, in part because their dynamics are closely connected with a number of fundamental nonlinear wave phenomena, including spontaneous modulational instability and rogue-wave formation \cite{SRADE2024,GA,GAZ}.

For the KdV equation,
\begin{align}\label{KdV}
	u_{t}+u_{xxx}-6uu_{x}=0,\ \ -\infty<x<+\infty,\ t>0,
\end{align}
classical theory provides two major exactly solvable settings: rapidly decaying data, treated by the inverse scattering transform \cite{GGKM}, and periodic or quasi-periodic finite-gap solutions, described by algebro-geometric methods \cite{Gesztesy,GengZhaiDai,WeiGeng}. A central difficulty arises when one considers bounded data that are neither decaying nor periodic. To construct new and broad families of solutions for integrable nonlinear PDEs,
the dressing method developed by Zakharov and Manakov has led to several notable results \cite{ZM}. In \cite{DZZ}, that approach produces a new class of solutions of the  KdV equation, called
\emph{primitive potentials},  which have the same spectral structure as periodic finite-gap potentials, but that are neither periodic nor
quasi-periodic. The primitive potentials, although not random, can be naturally
interpreted as  a useful model for integrable turbulence and  soliton gas.
These potentials are described by a pair of positive H\"older-continuous functions on the allowed bands. More precisely, the construction is formulated through a Riemann-Hilbert (RH) problem
for a vector $
\Phi=
(
	\Phi_1\
	\Phi_2
)^{\top},
$
which is normalized at infinity and satisfies the jump conditions
\begin{align}
\Phi_-(iz)=J(z)\Phi_+(iz),\qquad
\Phi_-(-iz)=J(z)^{\top}\Phi_+(-iz),\qquad z\in(a_1,b_1),
\end{align}
where $\Phi_+$ and $\Phi_-$ denote the left and right
values of $\Phi$ on the contours.
Here the jump matrix is
\begin{align}
J(z)=\frac{1}{1+r_1(z)r_2(z)}
\begin{pmatrix}
	1-r_1(z)r_2(z) & 2i\,r_1(z)e^{-2zx}\\
	2i\,r_2(z)e^{2zx} & 1-r_1(z)r_2(z)
\end{pmatrix}.
\end{align}
The parameters \(a_1\) and \(b_1\) are real and satisfy
$
0<a_1<b_1.
$ The reflection coefficients \(r_1(z;t)\) and \(r_2(z;t)\) evolve in time according to
\[
r_1(z;t)=r_1(z;0)e^{8z^3t},\qquad
r_2(z;t)=r_2(z;0)e^{-8z^3t}.
\]
Then Nabelek prove that all algebro-geometric finite-gap solutions of the KdV equation can be realized within the primitive potentials framework\cite{Nab}. Grava, Jenkins,Zhang and Zhang developed the direct and inverse scattering theory for the focusing NLS equation with initial data asymptotic to an elliptic background, and further related this class of data to full soliton gas initial data \cite{GJZZ2}.
Recently, Girotti et al. conducted a comprehensive study on the large-space and large-time asymptotics of the KdV soliton gas in the case $r_2\equiv 0$ \cite{Girotti-1}. By contrast, the corresponding large-time asymptotic problem in the  case \(r_1\neq 0\) and \(r_2\neq 0\) appears to be considerably more difficult and, to the best of our knowledge, remains open.
The purpose of the present paper is to address this problem.  We consider jump contours of the form
$i(a_1,b_1)\cup i(-b_1,-a_1)$
and allow the jump matrix on each contour component to involve two nonzero coefficients \(\hat r_1(k)\) and \(\hat \rho_1(k)\). Our starting point is a pure-soliton RH problem with \(2N\) poles and two different types of residue conditions, where \(N\in\mathbb N_+\). We then show that, as \(N\to\infty\), this problem converges to a continuous RH problem introduced by \cite{DZZ}. After this continuum limit is established, we analyze the resulting RH problem by means of the Deift--Zhou nonlinear steepest descent method \cite{DeiftZ,spin}, combined with a suitable $g$-function mechanism \cite{BWYZ,DeiftItsZhou,Girotti-2,GYJ,BL,HXF,Wang,TM,Fan,Ling,Yan,odd,fnls}. In this way, we derive an explicit large-time asymptotic formula for the new KdV soliton gas.

Our main result is the following theorem.
\begin{theorem}[Large-time asymptotics in the three regions]\label{thm-main-asymptotics}
Assume that $r_1(k)$ and $\rho_1(k)$ are positive and nonvanishing on $i[a_1,b_1]$, where $0<a_1<b_1$, and admit analytic extensions to a neighbourhood of $i[a_1,b_1]$ (cf.~Assumption~\ref{assumption1}).
Then, as $t\to+\infty$, the KdV potential $u(x,t)$ admits the following asymptotic descriptions.
	
	\begin{enumerate}
		\item[(i)] \textbf{Supercritical region:} If $\xi>\xi_*$, where $\xi_{*}$  is defined by \eqref{xicrit}, then
		\begin{equation}\label{eq:uasy-super}
			u(x,t)=b_1^2-a_1^2
			-2b_1^2\,\operatorname{dn}^2\!\left(
			b_1\bigl(x-2(a_1^2+b_1^2)t+\phi_{\mathrm{sup}}\bigr)+K(m)
			\,\middle|\, m
			\right)
			+\mathcal{O}(t^{-1}),
		\end{equation}
		where \(K(m)\) denotes the complete elliptic integrals of the first  kinds, $\operatorname{dn}(z|m)$ is the Jacobi elliptic function of modulus $m=\frac{a_1}{b_1}$, and the the phase shift is
		\begin{equation}\label{eq:phi-super}
			\phi_{\mathrm{sup}}
			=
			\frac{i}{\pi}\int_{a_1}^{b_1}
			\frac{\log \dfrac{2\hat{\rho}(\zeta)}{1+\hat{r}(\zeta)\hat{\rho}(\zeta)}}{R_+(\zeta)}\,d\zeta.
		\end{equation}
		
		\item[(ii)] \textbf{Intermediate region:} If $\xi_{\mathrm{crit}}<\xi<\xi_*$, where $\xi_{\mathrm{crit}}$  is defined by \eqref{xicrit}, then
		\begin{equation}\label{eq:uasy-mid}
			u(x,t)=b_1^2-a_1^2
			-2b_1^2\,\operatorname{dn}^2\!\left(
			b_1\bigl(x-2(a_1^2+b_1^2)t+\phi_{\mathrm{mid}}\bigr)+K(m)
			\,\middle|\, m
			\right)
			+\mathcal{O}(t^{-1/2}),
		\end{equation}
		where
		\begin{equation}\label{eq:phi-mid}
			\phi_{\mathrm{mid}}
			=
			\frac{i}{\pi}
			\left(
			\int_{a_1}^{\tilde\alpha}
			\frac{\log \dfrac{2\hat\rho(\zeta)}{1+\hat r(\zeta)\hat\rho(\zeta)}}{R_+(\zeta)}\,d\zeta
			-
			\int_{\tilde\alpha}^{b_1}
			\frac{\log \dfrac{2\hat r(\zeta)}{1+\hat r(\zeta)\hat\rho(\zeta)}}{R_+(\zeta)}\,d\zeta
			\right).
		\end{equation}
		
		\item[(iii)] \textbf{Subcritical region:} If $\xi<\xi_{\mathrm{crit}}$, then
		\begin{equation}\label{eq:uasy-sub}
			u(x,t)=b_1^2-a_1^2
			-2b_1^2\,\operatorname{dn}^2\!\left(
			b_1\bigl(x-2(a_1^2+b_1^2)t+\phi_{\mathrm{sub}}\bigr)+K(m)
			\,\middle|\, m
			\right)
			+\mathcal{O}(t^{-1}),
		\end{equation}
		where
		\begin{equation}\label{eq:phi-sub}
			\phi_{\mathrm{sub}}
			=
			-\frac{i}{\pi}
			\int_{a_1}^{b_1}
			\frac{\log \dfrac{2\hat r(\zeta)}{1+\hat r(\zeta)\hat\rho(\zeta)}}{R_+(\zeta)}\,d\zeta.
		\end{equation}
	\end{enumerate}
\end{theorem}

The paper is organized as follows. In Section \ref{sec2}, we formulate a pure-soliton RH problem for the KdV equation with two types of residue conditions, whose discrete spectrum is supported on
$
i(-b_1,-a_1)\cup i(a_1,b_1).
$
By passing to the limit \(N\to+\infty\), we obtain the RH problem for a new KdV soliton gas, in which the jump matrix carries two nonzero reflection coefficients on each contour. In Section \ref{sec3}, we construct the relevant \(g\)-function and perform a sequence of transformations reducing the original RH problem to explicitly solvable model problems. These model problems are solved in terms of Abelian integrals and Jacobi elliptic function on the associated hyperelliptic Riemann surface. This yields the theorem \ref{thm-main-asymptotics}.

\section{A RH characterization of the KdV equation}\label{sec2}
	The KdV equation \eqref{KdV} admits the Lax pair
	\begin{equation}\label{lax}
		\begin{aligned}
			&-\tilde{\Phi}_{xx}+u\tilde{\Phi}=k\tilde{\Phi},\\
			&\tilde{\Phi}_t-4\tilde{\Phi}_{xxx}+6u\tilde{\Phi}_x+3u_x\tilde{\Phi}=0,
		\end{aligned}
	\end{equation}
where $\tilde{\Phi}=\tilde{\Phi}(k;x,t)$  and the spectral parameter  $k \in\C$.
The RH problem for the pure soliton solution of the KdV equation (see example \cite{Girotti-1,Wang}) is described as follows:
\begin{RHP}\label{RHP1}
Find  a $1\times2$ vector-valued function $M(k;x,t)$ with the following properties
\begin{enumerate}
\item $M(k;x,t)$ is meromorphic in $\mathbb{C}$, with simple poles at $i\kappa_{j}, iz_j $ in $i\mathbb{R}_{+}$, and at the corresponding conjugate points $-i\kappa_{j}, -iz_j $ in $i\mathbb{R}_{-}$,\ ${j=1,2,\ldots,N},\   N\in \Z_+$.
			
\item $M(k;x,t)$ satisfies the residue conditions
		\begin{equation}
			\begin{aligned}\label{residue_soliton}
\Res\limits_{k=i\kappa_{j}}{M}(k)&=\lim_{k\to i\kappa_{j}}{M}(k)\begin{pmatrix}0&0\\i\chi_{j}e^{2i\theta(k;x,t)}&0\\\end{pmatrix},\\
\Res\limits_{k=iz_{j}}{M}(k)&=\lim_{k\to iz_{j}}{M}(k)\begin{pmatrix}0&i\mu_{j}e^{-2i\theta(k;x,t)}\\ 0&0\\\end{pmatrix},\\
\Res\limits_{k=-i\kappa_j}{M}(k)&
=\lim_{k\to-i\kappa_j}{M}(k)\begin{pmatrix}0&-i\chi_je^{-2i\theta(k;x,t)}\\0&0\end{pmatrix},\\
\Res\limits_{k=-iz_j}{M}(k)&
=\lim_{k\to-iz_j}{M}(k)\begin{pmatrix}0&0\\ -i\mu_je^{2i\theta(k;x,t)}&0\end{pmatrix},
			\end{aligned}
		\end{equation}
		where $\theta(k;x,t)=4tk^3+xk$ and $\chi_{j}, \mu_j$ are nonzero, positive and real constant.
\item $M(k;x,t)\rightarrow \begin{pmatrix}
	1&1
\end{pmatrix} ,\  k\rightarrow\infty$.
\item $M(k;x,t)$ satisfies the symmetry
		\[
		\ M(-k)=M(k)\begin{pmatrix}
			0&1\\1&0\end{pmatrix}\ .
		\]
\end{enumerate}
\end{RHP}
The potential $u(x,t)$ is determined from $M(k;x,t)$ via
	\begin{equation}
		u(x,t)=2\dfrac{\d}{\d x}\lim\limits_{k\to\infty}\le(\frac{k}{i} (M_1(k;x,t)-1)\ri),
	\end{equation}
	where $M_1(k)$ is the first entry of the vector $M(k)$.
\begin{assumption}\label{assumption1}
	As $N\to\infty$, we assume that:
	\begin{enumerate}
		\item
		The poles $\{i\kappa_j\}_{j=1}^N$ and $\{iz_j\}_{j=1}^N$
		are uniformly distributed and interlaced on $i[a_1,b_1]$, where
		$0<a_1<b_1$, namely
		\[
		\kappa_j=a_1+j\frac{b_1-a_1}{N},
		\qquad
		z_j=a_1+\left(j-\frac12\right)\frac{b_1-a_1}{N},
		\qquad j=1,\dots,N.
		\]
		
		\item
		The coefficients $\{i\chi_j\}_{j=1}^N$ and $\{i\mu_j\}_{j=1}^N$
		are purely imaginary and are discretizations of two given analytic functions:
		\begin{gather}
			i\chi_j
			=
			i\frac{b_1-a_1}{2N}r_1(i\kappa_j)
			\prod_{\substack{m=1\\ m\neq j}}^N
			\frac{\kappa_j-z_m}{\kappa_j-\kappa_m},
			\qquad j=1,\dots,N,
			\\
			i\mu_j
			=
			i\frac{b_1-a_1}{2N}\rho_1(iz_j)
			\prod_{\substack{m=1\\ m\neq j}}^N
			\frac{z_j-\kappa_m}{z_j-z_m},
			\qquad j=1,\dots,N,
		\end{gather}
		where $r_1(k)$ and $\rho_1(k)$ are analytic function in a neighborhood of the intervals
		$i[a_1,b_1]\cup i[-b_1,-a_1]$, satisfy
		\[
		r_1(k)\rho_1(k)< 1,
		\qquad
		r_1(-k)=r_1(k),
		\qquad
		\rho_1(-k)=\rho_1(k),
		\]
		and is further assumed to be a real valued positive
		and non-vanishing function of $k$ for $k\in i[a_1,b_1]$.
	\end{enumerate}
\end{assumption}
	We introduce the closed curves $\Gamma_{1+}$ which located in the upper half-plane $\mathbb{C}_{+}$, encircling the interval $i[a_1,b_1]$ in a counterclockwise direction. Similarly, $\Gamma_{1-}$ are the counterclockwise curves that surrounds the interval $i[-b_1,-a_1]$ in the lower half-plane $\mathbb{C}_{-}$.
We first remove the poles by defining
\begin{equation}\label{Zdef}
	Z(k)=M(k)
	\begin{cases}
		\begin{pmatrix}
			1&0\\
			-r_1(k)e^{2i\theta(k)}\displaystyle\prod_{j=1}^N\frac{k-iz_j}{k-i\kappa_j}&1
		\end{pmatrix}
		\begin{pmatrix}
			1&\rho_1(k)e^{-2i\theta(k)}\displaystyle\prod_{j=1}^N\frac{k-i\kappa_j}{k-iz_j}\\
			0&1
		\end{pmatrix},
		& k \text{ inside }\Gamma_{1+},
		\\[4ex]
		\begin{pmatrix}
			1&-r_1(k)e^{-2i\theta(k)}\displaystyle\prod_{j=1}^N\frac{k+iz_j}{k+i\kappa_j}\\
			0&1
		\end{pmatrix}
		\begin{pmatrix}
			1&0\\
			\rho_1(k)e^{2i\theta(k)}\displaystyle\prod_{j=1}^N\frac{k+i\kappa_j}{k+iz_j}&1
		\end{pmatrix},
		& k \text{ inside }\Gamma_{1-},
		\\[3ex]
		I,
		& \text{otherwise},
	\end{cases}
\end{equation}
where $I$ is the $2\times 2$ identity matrix.

With this choice, all poles of $M$ are removed, and the $1\times 2$ row-vector-valued function $Z(k;x,t)$ satisfies the jump relation
\begin{equation}\label{Zjump-discrete}
	Z_+(k)=Z_-(k)
	\begin{cases}
		\begin{pmatrix}
			1&0\\
			-r_1(k)e^{2i\theta(k)}\displaystyle\prod_{j=1}^N\frac{k-iz_j}{k-i\kappa_j}&1
		\end{pmatrix}
		\begin{pmatrix}
			1&\rho_1(k)e^{-2i\theta(k)}\displaystyle\prod_{j=1}^N\frac{k-i\kappa_j}{k-iz_j}\\
			0&1
		\end{pmatrix},
		& k\in \Gamma_{1+},
		\\[4ex]
		\begin{pmatrix}
			1&-r_1(k)e^{-2i\theta(k)}\displaystyle\prod_{j=1}^N\frac{k+iz_j}{k+i\kappa_j}\\
			0&1
		\end{pmatrix}
		\begin{pmatrix}
			1&0\\
			\rho_1(k)e^{2i\theta(k)}\displaystyle\prod_{j=1}^N\frac{k+i\kappa_j}{k+iz_j}&1
		\end{pmatrix},
		& k\in \Gamma_{1-}.
	\end{cases}
\end{equation}
Here, on $\Gamma_{1+}\cup\Gamma_{1-}$, the boundary value $Z_+(k)$ is taken from the left side of the contour and $Z_-(k)$ from the right side.

To pass to the continuum limit, set
\[
B_N(k):=\prod_{j=1}^N\frac{k-iz_j}{k-i\kappa_j},
\qquad
\widetilde B_N(k):=\prod_{j=1}^N\frac{k+iz_j}{k+i\kappa_j}.
\]

\begin{proposition}\label{prop:BNlimit}
	As $N\to\infty$, one has
	\begin{gather}
		B_N(k)\longrightarrow \beta(k):=
		\left(\frac{k-ia_1}{k-ib_1}\right)^{1/2},
		\qquad
		k\in \C\setminus i[a_1,b_1],\label{limit1-new}
		\\
		\widetilde B_N(k)\longrightarrow \widetilde\beta(k):=
		\left(\frac{k+ia_1}{k+ib_1}\right)^{1/2},
		\qquad
		k\in \C\setminus i[-b_1,-a_1],\label{limit2-new}
	\end{gather}
	uniformly on compact subsets of the indicated domains. The branches of $\beta$ and $\widetilde\beta$ are fixed by the normalization
	\[
	\beta(k)\to 1,
	\qquad
	\widetilde\beta(k)\to 1,
	\qquad
	k\to\infty.
	\]
\end{proposition}

\begin{proof}
	Let
	\[
	\delta=\frac{b_1-a_1}{N},
	\qquad
	\kappa_j=a_1+j\delta,
	\qquad
	z_j=a_1+\left(j-\frac12\right)\delta.
	\]
	Then
	\[
	B_N(k)=\prod_{j=1}^N
	\frac{k-i(a_1+(j-\frac12)\delta)}{k-i(a_1+j\delta)}.
	\]
	Hence
	\[
	\log B_N(k)
	=
	\sum_{j=1}^N
	\log\left(
	1+\frac{i\delta/2}{k-i(a_1+j\delta)}
	\right).
	\]
	For $k$ in any compact subset of $\C\setminus i[a_1,b_1]$, we have
	\[
	\log(1+w)=w+\mathcal{O}(w^2),
	\]
	uniformly for small $w$, and therefore
	\[
	\log B_N(k)
	=
	\frac12\sum_{j=1}^N
	\frac{i\delta}{k-i(a_1+j\delta)}
	+\mathcal{O}(N\delta^2).
	\]
	Since $N\delta^2=(b_1-a_1)^2/N\to 0$, the Riemann sum converges uniformly on compact sets to
	\[
	\frac12\int_{a_1}^{b_1}\frac{i\,d\zeta}{k-i\zeta}.
	\]
	Thus
	\[
	\log B_N(k)\to
	\frac12\int_{a_1}^{b_1}\frac{i\,d\zeta}{k-i\zeta}
	=
	\frac12\log\left(\frac{k-ia_1}{k-ib_1}\right),
	\]
	which proves \eqref{limit1-new}. The proof of \eqref{limit2-new} is analogous.
\end{proof}

By Proposition~\ref{prop:BNlimit}, the jump matrix on $\Gamma_{1\pm}$ converges uniformly to its continuum counterpart. Therefore, in the limit $N\to\infty$, \eqref{Zjump-discrete} becomes
\begin{equation}\label{Zjump-limit}
	Z_+(k)=Z_-(k)
	\begin{cases}
		\begin{pmatrix}
			1&0\\
			-r_1(k)e^{2i\theta(k)}\beta(k)&1
		\end{pmatrix}
		\begin{pmatrix}
			1&\rho_1(k)e^{-2i\theta(k)}\beta(k)^{-1}\\
			0&1
		\end{pmatrix},
		& k\in\Gamma_{1+},
		\\[4ex]
		\begin{pmatrix}
			1&-r_1(k)e^{-2i\theta(k)}\widetilde\beta(k)\\
			0&1
		\end{pmatrix}
		\begin{pmatrix}
			1&0\\
			\rho_1(k)e^{2i\theta(k)}\widetilde\beta(k)^{-1}&1
		\end{pmatrix},
		& k\in\Gamma_{1-}.
	\end{cases}
\end{equation}

Next, we define
\begin{equation}\label{Xdef}
	X(k)=Z(k)
	\begin{cases}
		\begin{pmatrix}
			1&-\rho_1(k)e^{-2i\theta(k)}\beta(k)^{-1}\\
			0&1
		\end{pmatrix}
		\begin{pmatrix}
			1&0\\
			r_1(k)e^{2i\theta(k)}\beta(k)&1
		\end{pmatrix},
		& k \text{ inside }\Gamma_{1+},
		\\[4ex]
		\begin{pmatrix}
			1&0\\
			-\rho_1(k)e^{2i\theta(k)}\widetilde\beta(k)^{-1}&1
		\end{pmatrix}
		\begin{pmatrix}
			1&r_1(k)e^{-2i\theta(k)}\widetilde\beta(k)\\
			0&1
		\end{pmatrix},
		& k \text{ inside }\Gamma_{1-},
		\\[3ex]
		I,
		& \text{otherwise}.
	\end{cases}
\end{equation}
By construction, the jumps of $X$ on $\Gamma_{1+}$ and $\Gamma_{1-}$ are removed.

The remaining jumps are located on the branch cuts $i(a_1,b_1)$ and $i(-b_1,-a_1)$. Since
\begin{align}
	\beta_+(k)=-\beta_-(k), \qquad & k\in i(a_1,b_1),\\
	\widetilde\beta_+(k)=-\widetilde\beta_-(k), \qquad & k\in i(-b_1,-a_1),
\end{align}
where both cuts are oriented upward, we obtain the following RH problem.

\begin{RHP}\label{RHP2}
	Find a $1\times 2$ vector-valued function $X(k;x,t)$ such that:
	\begin{enumerate}
		\item
		$X(k;x,t)$ is analytic for
		\[
		k\in \C\setminus\bigl(i[a_1,b_1]\cup i[-b_1,-a_1]\bigr).
		\]
		
		\item
		For $k\in i(a_1,b_1)\cup i(-b_1,-a_1)$, the boundary values $X_{\pm}(k)$ are taken from the left/right side of the oriented contour, and
		\begin{equation}\label{Xjump-factorized}
			X_+(k)=X_-(k)V_X(k),
		\end{equation}
		where
		\begin{equation}\label{VX-factorized}
			V_X(k)=
			\begin{cases}
				\begin{pmatrix}
					1&0\\
				\scriptstyle	r_1(k)\beta_+(k)e^{2i\theta(k)}&1
				\end{pmatrix}
				\begin{pmatrix}
					1& \scriptstyle-2\rho_1(k)\beta_+(k)^{-1}e^{-2i\theta(k)}\\
					0&1
				\end{pmatrix}
				\begin{pmatrix}
					1&0\\
					\scriptstyle r_1(k)\beta_+(k)e^{2i\theta(k)}&1
				\end{pmatrix},
				& k\in i(a_1,b_1),
				\\[5ex]
				\begin{pmatrix}
					1& \scriptstyle r_1(k)\widetilde\beta_+(k)e^{-2i\theta(k)}\\
					0&1
				\end{pmatrix}
				\begin{pmatrix}
					1&0\\
					\scriptstyle -2\rho_1(k)\widetilde\beta_+(k)^{-1}e^{2i\theta(k)}&1
				\end{pmatrix}
				\begin{pmatrix}
					1& \scriptstyle r_1(k)\widetilde\beta_+(k)e^{-2i\theta(k)}\\
					0&1
				\end{pmatrix},
				& k\in i(-b_1,-a_1).
			\end{cases}
		\end{equation}
		
		\item
		As $k\to\infty$,
		\[
		X(k;x,t)=\begin{pmatrix}1&1\end{pmatrix}+\mathcal{O}(k^{-1}).
		\]
	\end{enumerate}
\end{RHP}
Moreover, by using the symmetries
\[
r_1(-k)=r_1(k), \qquad \rho_1(-k)=\rho_1(k),
\]
one verifies that
\[
V_X(-k)=\sigma_1V_X(k)\sigma_1,\ \sigma_1=\begin{pmatrix}
	0&1\\1&0
\end{pmatrix},
\]
and hence
\[
X(-k)=X(k)\sigma_1.
\]
Since $\beta_+(k)\in -i\mathbb R_+$ for $k\in i(a_1,b_1)$, we introduce
\[
-i\hat r_1(k)=r_1(k)\beta_+(k),
\qquad
i\frac{\hat\rho_1(k)}{1+\hat r_1(k)\hat\rho_1(k)}
=\rho_1(k)\beta_+^{-1}(k).
\]
Then $\hat r_1(k)>0$, and
\[
\frac{\hat\rho_1(k)}{1+\hat r_1(k)\hat\rho_1(k)}>0.
\]
If in addition $r_1(k)\rho_1(k)<1$, then $\hat\rho_1(k)>0$.
Then the jump matrix $V_X(k)$ for $X(k)$ transform into \begin{align}\label{factor1}
	& V_X(k)= \begin{cases} \begin{pmatrix} {1}& 0\\ {-i\hat{r}_1(k) e^{2 i\theta(k)} } & {1} \end{pmatrix} \begin{pmatrix} {1}& {-2i \frac{\hat{\rho}_1(k)}{1+\hat{r}_1(k)\hat{\rho}_1(k)}e^{-2 i\theta(k)} }\\0 & {1} \end{pmatrix} \begin{pmatrix} {1}& 0\\ {-i\hat{r}_1(k) e^{2 i\theta(k)} } & {1} \end{pmatrix}, & k \in i(a_1,b_1),\\ \begin{pmatrix} 1 & {i\hat{r}_1(-k)}{e^{-2i\theta(k)}} \\ 0 & 1\end{pmatrix} \begin{pmatrix} 1 & 0\\ 2i\frac{\hat{\rho}_1(-k)}{1+\hat{r}_1(-k)\hat{\rho}_1(-k)}e^{2i\theta(k)} & 1\end{pmatrix} \begin{pmatrix} 1 & {i\hat{r}_1(-k)}{e^{-2i\theta(k)}} \\ 0 & 1\end{pmatrix}, & k \in i(-b_1,-a_1). \end{cases}
	\end{align}
	A direct computation shows that
	\begin{align}
		& V_X(k)= \begin{cases} \begin{pmatrix} \dfrac{1-\hat{r}_1(k)\hat{\rho}_1(k)}{1+\hat{r}_1(k)\hat{\rho}_1(k)}& -\dfrac{2i\hat{\rho}_1(k)}{1+\hat{r}_1(k)\hat{\rho}_1(k)}e^{-2 i\theta(k)}\\ -{\dfrac{2i\hat{r}_1(k)}{1+\hat{r}_1(k)\hat{\rho}_1(k)} e^{2 i\theta(k)} } & \dfrac{1-\hat{r}_1(k)\hat{\rho}_1(k)}{1+\hat{r}_1(k)\hat{\rho}_1(k)} \end{pmatrix}, & k \in i(a_1,b_1),\\ \begin{pmatrix} \dfrac{1-\hat{r}_1(-k)\hat{\rho}_1(-k)}{1+\hat{r}_1(-k)\hat{\rho}_1(-k)} & 2i\frac{\hat{r}_1(-k)}{1+\hat{r}_1(-k)\hat{\rho}_1(-k)}e^{-2i\theta(k)}\\ 2i\frac{\hat{\rho}_1(-k)}{1+\hat{r}_1(-k)\hat{\rho}_1(-k)}e^{2i\theta(k)} &\dfrac{1-\hat{r}_1(-k)\hat{\rho}_1(-k)}{1+\hat{r}_1(-k)\hat{\rho}_1(-k)} \end{pmatrix}, & k \in i(-b_1,-a_1). \end{cases}
		\end{align}
		It is easy to check that the RH problem for $X(k)$ has the same form of the RH problem for KdV primitive potentials introduced in \cite{DZZ}. The jump matrices of above form also appear in Ref.~\cite{GJZZ,Zhu}.
Finally, since all transformations are normalized by the identity at infinity, the potential is recovered from $X$ by
\begin{equation}
	u(x,t)=2\frac{d}{dx}\lim_{k\to\infty}
	\left(
	\frac{k}{i}\bigl(X_1(k;x,t)-1\bigr)
	\right).
\end{equation}

\section{Behavior of the potential $u(x,t)$ as $t\rightarrow +\infty$}\label{sec3}
We first transform the jump contours to the real line. We set $\Sigma_1=(a_1,b_1)$, $\Sigma_{-1}=(-b_1,-a_1)$ and
\begin{align}
	Y(k)=X(ik),\ \hat{r}_1(ik)=\hat{r}(k), \hat{\rho}_1(ik)=\hat{\rho}(k),\ t\hat{\theta}(k)=i\theta(ik)=t(4k^3-4k\xi), \ \xi=\frac{x}{4t}.
\end{align}
We get
\begin{align}
	Y_{+}(k)=Y_{-}(k)V_{Y}(k),\ V_{Y}(k)=\begin{cases}
		\begin{pmatrix} \dfrac{1-\hat{r}(k)\hat{\rho}(k)}{1+\hat{r}(k)\hat{\rho}(k)}&   -\dfrac{2i\hat{\rho}(k)}{1+\hat{r}(k)\hat{\rho}(k)}e^{-2 t\hat{\theta}(k)}\\   -{\dfrac{2i\hat{r}(k)}{1+\hat{r}(k)\hat{\rho}(k)} e^{2 t\hat{\theta}(k)} }  & \dfrac{1-\hat{r}(k)\hat{\rho}(k)}{1+\hat{r}(k)\hat{\rho}(k)} \end{pmatrix}, &  k \in \Sigma_1,\\
		\begin{pmatrix} \dfrac{1-\hat{r}(-k)\hat{\rho}(-k)}{1+\hat{r}(-k)\hat{\rho}(-k)} & 2i\frac{\hat{r}(-k)}{1+\hat{r}(-k)\hat{\rho}(-k)}e^{-2t\hat{\theta}(k)}\\  2i\frac{\hat{\rho}(-k)}{1+\hat{r}(-k)\hat{\rho}(-k)}e^{2t\hat{\theta}(k)} &\dfrac{1-\hat{r}(-k)\hat{\rho}(-k)}{1+\hat{r}(-k)\hat{\rho}(-k)}
		\end{pmatrix}, &   k \in \Sigma_{-1}.
	\end{cases}
\end{align}
Since $\hat{r}(k), \hat{\rho}(k)\neq 0$, one of the off-diagonal term of the jump matrix $V_{Y}(k)$ is always growing. To overcome this problem, we first introduce two new triangle factorization for  $V_{Y}(k)$ on $\Sigma_{1}$:
\begin{align}
		V_{Y}(k)&=\begin{pmatrix} {1}& 0\\  \scriptstyle {-i  \frac{\hat{r}(k)\hat{\rho}(k)-1}{2\hat{\rho}(k)} e^{2 t\hat{\theta}(k)} }  & {1} \end{pmatrix}
		\begin{pmatrix} 0 &  \scriptstyle {-i  \frac{2\hat{\rho}(k)}{1+\hat{r}(k)\hat{\rho}(k)}e^{-2 t\hat{\theta}(k)} }\\ \scriptstyle -i  \le(\frac{2\hat{\rho}(k)}{1+\hat{r}(k)\hat{\rho}(k)}\ri)^{-1}e^{2 t\hat{\theta}(k)}  & 0 \end{pmatrix}
		\begin{pmatrix} {1}& 0\\ \scriptstyle {-i  \frac{\hat{r}(k)\hat{\rho}(k)-1}{2\hat{\rho}(k)} e^{2 t\hat{\theta}(k)} } & {1} \end{pmatrix},\label{factor11}\\
		&=\begin{pmatrix} 1 &   \scriptstyle -i\frac{\hat{r}(k)\hat{\rho}(k)-1}{2\hat{r}(k)}e^{-2t\hat{\theta}(k)} \\ 0 & 1\end{pmatrix}
		\begin{pmatrix} 0 & \scriptstyle -i\le(\frac{2\hat{r}(k)}{1+\hat{r}(k)\hat{\rho}(k)}\ri)^{-1}e^{-2t\hat{\theta}(k)}\\  \scriptstyle -i\frac{2\hat{r}(k)}{1+\hat{r}(k)\hat{\rho}(k)}e^{2t\hat{\theta}(k)} & 0\end{pmatrix}	\begin{pmatrix} 1 &  \scriptstyle -i\frac{\hat{r}(k)\hat{\rho}(k)-1}{2\hat{r}(k)}e^{-2t\hat{\theta}(k)} \\ 0 & 1\end{pmatrix},\label{factor12}
\end{align}
and on $\Sigma_{-1}$:
\begin{align}
		V_{Y}(k)&=\begin{pmatrix} 1 &   \scriptstyle i\frac{\hat{r}(-k)\hat{\rho}(-k)-1}{2\hat{\rho}(-k)}e^{-2t\hat{\theta}(k)} \\ 0 & 1\end{pmatrix}
	\begin{pmatrix} 0 & \scriptstyle i\le(\frac{2\hat{\rho}(-k)}{1+\hat{r}(-k)\hat{\rho}(-k)}\ri)^{-1}e^{-2t\hat{\theta}(k)}\\  \scriptstyle i\frac{2\hat{\rho}(-k)}{1+\hat{r}(-k)\hat{\rho}(-k)}e^{2t\hat{\theta}(k)} & 0\end{pmatrix}	\begin{pmatrix} 1 &  \scriptstyle i\frac{\hat{r}(-k)\hat{\rho}(-k)-1}{2\hat{\rho}(-k)}e^{-2t\hat{\theta}(k)} \\ 0 & 1\end{pmatrix},\label{factor21}\\
	&=\begin{pmatrix} {1}& 0\\  \scriptstyle {i  \frac{\hat{r}(-k)\hat{\rho}(-k)-1}{2\hat{r}(-k)} e^{2 t\hat{\theta}(k)} }  & {1} \end{pmatrix}
	\begin{pmatrix} 0 &  \scriptstyle {i  \frac{2\hat{r}(-k)}{1+\hat{r}(-k)\hat{\rho}(-k)}e^{-2 t\hat{\theta}(k)} }\\ \scriptstyle i  \le(\frac{2\hat{r}(-k)}{1+\hat{r}(-k)\hat{\rho}(-k)}\ri)^{-1}e^{2t\hat{\theta}(k)}  & 0 \end{pmatrix}
	\begin{pmatrix} {1}& 0\\ \scriptstyle {i  \frac{\hat{r}(-k)\hat{\rho}(-k)-1}{2\hat{r}(-k)} e^{2 t\hat{\theta}(k)} } & {1} \end{pmatrix}.\label{factor22}
\end{align}
Next, we introduce a $g$-function
 by prescribing its derivative as
\begin{equation}\label{gprime}
	g'(k)
	=
	-12k^2+4\xi
	+
	\frac{12Q_2(k)-4\xi Q_1(k)}{R(k)},
\end{equation}
where
\begin{equation}\label{RQ}
	R(k)=\sqrt{(k^2-a_1^2)(k^2-b_1^2)},
	\qquad
	Q_1(k)=k^2+c_1,
	\qquad
	Q_2(k)=k^4-\frac12(a_1^2+b_1^2)k^2+c_2.
\end{equation}
The constants \(c_1\) and \(c_2\) are determined by the conditions
\begin{equation}\label{moments}
	\int_0^{a_1}\frac{Q_1(\zeta)}{R_+(\zeta)}\,d\zeta=0,
	\qquad
	\int_0^{a_1}\frac{Q_2(\zeta)}{R_+(\zeta)}\,d\zeta=0.
\end{equation}
Integrating \eqref{gprime} from \(b_1\) to \(k\), we obtain
\begin{equation}\label{eq:g}
	g(k)
	=
	-4k^3+4\xi k
	+\int_{b_1}^{k}
	\frac{12Q_2(\zeta)-4\xi Q_1(\zeta)}{R(\zeta)}\,d\zeta.
\end{equation}

By construction, \(g(k)\) is analytic in \(\mathbb{C}\setminus[-b_1,b_1]\) and satisfies
the jump conditions
\begin{equation}\label{g-jump1}
	g_+(k)+g_-(k)+8k^3-8\xi k=0,
	\qquad
	k\in (a_1,b_1)\cup(-b_1,-a_1),
\end{equation}
\begin{equation}\label{g-jump2}
	g_+(k)-g_-(k)=\Omega,
	\qquad
	k\in[-a_1,a_1],
\end{equation}
\begin{equation}\label{g-inf}
	g(k)=\mathcal{O}\!\left(\frac{1}{k}\right),
	\qquad k\to\infty.
\end{equation}
Here
\begin{equation}\label{eq:subcritical-Omega}
	\Omega
	=
	\frac{2\pi i\,b_1}{K(m)}
	\bigl(2\xi-a_1^2-b_1^2\bigr)
	\in i\mathbb{R},
	\qquad
	m=\frac{a_1}{b_1}.
\end{equation}
To  open the lenses, one has to determine the sign of
\[
\Re\!\ \bigl(2g(k)+8k^3-8\xi k\bigr)
\]
on the lens boundaries.

\begin{proposition}\label{prop:q-roots}
	Let
	\[
	q(r;\xi)
	=
	12\left(r^2-\frac12(a_1^2+b_1^2)r+c_2\right)-4\xi(r+c_1),
	\]
	where
	\[
	c_1=-b_1^2+b_1^2\frac{E(m)}{K(m)},
	\qquad
	c_2=\frac13 a_1^2 b_1^2+\frac16(a_1^2+b_1^2)c_1,
	\qquad
	m=\frac{a_1}{b_1}\in(0,1),
	\]
  and \(E(m)\) denotes the complete elliptic integrals of the
	second kinds.
	Define
	\[
	Q_1(k)=k^2+c_1,
	\qquad
	Q_2(k)=k^4-\frac12(a_1^2+b_1^2)k^2+c_2,
	\]
	and
	\begin{align}\label{xicrit}
	\xi_{\mathrm{crit}}=\frac{3Q_2(a_1)}{Q_1(a_1)},
	\qquad
	\xi_*=\frac{3Q_2(b_1)}{Q_1(b_1)}.
	\end{align}
	Then the roots of \(q(r;\xi)\) satisfy:
	
	\begin{enumerate}
		\item If \(\xi<\xi_{\mathrm{crit}}\), then \(q(\cdot;\xi)\) has one root in \((-\infty,a_1^2)\) and one root in \((0,a_1^2)\).
		
		\item If \(\xi_{\mathrm{crit}}<\xi<\xi_*\), then \(q(\cdot;\xi)\) has exactly one root in \((0,a_1^2)\) and exactly one root in \((a_1^2,b_1^2)\).
		
		\item If \(\xi>\xi_*\), then \(q(\cdot;\xi)\) has exactly one root in \((0,a_1^2)\) and exactly one root in \((b_1^2,\infty)\).
	\end{enumerate}
\end{proposition}

\begin{proof}
	Since
	\[
	c_1=b_1^2\left(\frac{E(m)}{K(m)}-1\right)<0,
	\]
	we have
	\[
	0<-c_1<a_1^2,
	\qquad
	a_1^2+c_1>0,
	\qquad
	b_1^2+c_1>0.
	\]
	Moreover,
	\[
	q(a_1^2;\xi)
	=
	12Q_2(a_1)-4\xi Q_1(a_1)
	=
	4Q_1(a_1)(\xi_{\mathrm{crit}}-\xi)
	=
	4(a_1^2+c_1)(\xi_{\mathrm{crit}}-\xi),
	\]
	and similarly,
	\[
	q(b_1^2;\xi)
	=
	12Q_2(b_1)-4\xi Q_1(b_1)
	=
	4Q_1(b_1)(\xi_*-\xi)
	=
	4(b_1^2+c_1)(\xi_*-\xi).
	\]
	Hence
	\[
	q(a_1^2;\xi)
	\begin{cases}
		>0,& \xi<\xi_{\mathrm{crit}},\\
		<0,& \xi>\xi_{\mathrm{crit}},
	\end{cases}
	\qquad
	q(b_1^2;\xi)
	\begin{cases}
		>0,& \xi<\xi_*,\\
		<0,& \xi>\xi_*.
	\end{cases}
	\]
	
	Next, at the fixed point \(r=-c_1\in(0,a_1^2)\), the \(\xi\)-dependence disappears:
	\[
	q(-c_1;\xi)
	=
	12\left(c_1^2+\frac12(a_1^2+b_1^2)c_1+c_2\right).
	\]
	Using the explicit formulas for \(c_1\) and \(c_2\), one verifies that
	\[
	q(-c_1;\xi)<0.
	\]
	
	We now distinguish three cases.
	
	\medskip
	\noindent
	\textit{Case 1: \(\xi<\xi_{\mathrm{crit}}\).}
	Since
	\[
	q(-c_1;\xi)<0<q(a_1^2;\xi),
	\qquad -c_1\in(0,a_1^2),
	\]
	the intermediate value theorem yields one root \(d_1\in(-c_1,a_1^2)\subset(0,a_1^2)\).
	Since \(q(\cdot;\xi)\) is a quadratic polynomial with positive leading coefficient, it has
	exactly two roots (counted with multiplicity). Let the other root be \(d_2\). If \(d_2>a_1^2\),
	then \(a_1^2\) lies between \(d_1\) and \(d_2\), and hence
	\[
	q(a_1^2;\xi)=12(a_1^2-d_1)(a_1^2-d_2)<0,
	\]
	which contradicts \(q(a_1^2;\xi)>0\). Therefore \(d_2<a_1^2\). Hence one root lies in
	\((0,a_1^2)\), while the other lies in \((-\infty,a_1^2)\).
	
	\medskip
	\noindent
	\textit{Case 2: \(\xi_{\mathrm{crit}}<\xi<\xi_*\).}
	At \(\xi=\xi_{\mathrm{crit}}\), we have \(q(a_1^2;\xi_{\mathrm{crit}})=0\). Since also
	\(q(-c_1;\xi_{\mathrm{crit}})<0\) and \(-c_1\in(0,a_1^2)\), there is another root in
	\((0,-c_1)\), and therefore
	\[
	q(0;\xi_{\mathrm{crit}})>0.
	\]
	Because
	\[
	q(0;\xi)=12c_2-4\xi c_1
	\]
	is strictly increasing in \(\xi\) (as \(c_1<0\)), it follows that
	\[
	q(0;\xi)>q(0;\xi_{\mathrm{crit}})>0
	\qquad\text{for all }\xi>\xi_{\mathrm{crit}}.
	\]
	Hence, for \(\xi_{\mathrm{crit}}<\xi<\xi_*\),
	\[
	q(0;\xi)>0,\qquad q(a_1^2;\xi)<0,\qquad q(b_1^2;\xi)>0.
	\]
	Therefore \(q(\cdot;\xi)\) has one root in \((0,a_1^2)\) and one root in \((a_1^2,b_1^2)\).
	
	\medskip
	\noindent
	\textit{Case 3: \(\xi>\xi_*\).}
	As above, \(q(0;\xi)>0\) for every \(\xi>\xi_{\mathrm{crit}}\). Thus
	\[
	q(0;\xi)>0,\qquad q(a_1^2;\xi)<0,\qquad q(b_1^2;\xi)<0.
	\]
	Hence \(q(\cdot;\xi)\) has one root in \((0,a_1^2)\). Since \(q(r;\xi)\to+\infty\) as
	\(r\to+\infty\), while \(q(b_1^2;\xi)<0\), it has another root in \((b_1^2,\infty)\).
	
	Finally, since \(q(\cdot;\xi)\) is a quadratic polynomial with positive leading coefficient,
	it can have at most two real roots. Therefore the above roots are exactly all the roots of
	\(q(\cdot;\xi)\).
\end{proof}

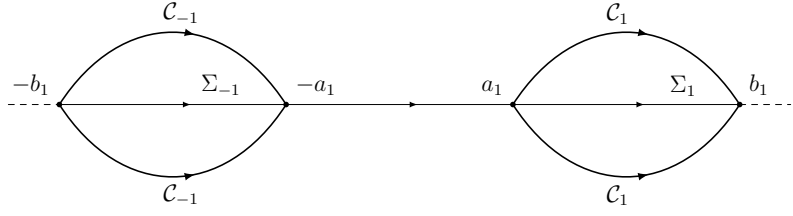
\begin{figure}
	\centering
	\scalebox{.75}{
		\begin{tikzpicture}[scale=1.0, line cap=round, line join=round]
			
			\tikzset{
				myarrow/.style={
					postaction={decorate},
					decoration={markings, mark=at position 0.58 with {\arrow{latex}}}
				}
			}
			
			\coordinate (L2) at (-6,0);
			\coordinate (L1) at (-2,0);
			\coordinate (R1) at ( 2,0);
			\coordinate (R2) at ( 6,0);
			
			\draw[dashed] (-6.9,0)--(L2);
			\draw[dashed] (R2)--(6.9,0);
			
			\draw[thin,myarrow] (L2)--(L1);
			\draw[thin,myarrow] (L1)--(R1);
			\draw[thin,myarrow] (R1)--(R2);
			
			\draw[thick,myarrow] (L2) .. controls (-4.9,1.7) and (-3.1,1.7) .. (L1);
			\draw[thick,myarrow] (R1) .. controls (3.1,1.7) and (4.9,1.7) .. (R2);
			
			\draw[thick,myarrow] (L2) .. controls (-4.9,-1.7) and (-3.1,-1.7) .. (L1);
			\draw[thick,myarrow] (R1) .. controls (3.1,-1.7) and (4.9,-1.7) .. (R2);
			
			\fill (L2) circle (1.4pt);
			\fill (L1) circle (1.4pt);
			\fill (R1) circle (1.4pt);
			\fill (R2) circle (1.4pt);
			
			\node[above left=1pt]  at (L2) {$-b_1$};
			\node[above right=1pt] at (L1) {$-a_1$};
			\node[above left=1pt]  at (R1) {$a_1$};
			\node[above right=1pt] at (R2) {$b_1$};
			
			\node at (-3.15, 0.33) {$\Sigma_{-1}$};
			\node at ( 5.00, 0.33) {$\Sigma_1$};
			\node at (-3.85,-1.60) {$\mathcal{C}_{-1}$};
			\node at ( 3.85,-1.60) {$\mathcal{C}_1$};
				\node at (-3.85,1.60) {$\mathcal{C}_{-1}$};
			\node at ( 3.85,1.60) {$\mathcal{C}_1$};
		\end{tikzpicture}
	}
	\caption{Opening lenses and the jump contours for $Z(k)$.}
	\label{openinglenses1}
\end{figure}

\begin{lemma}\label{lemma3.1}
	The following sign properties hold.
	
	\smallskip
	\noindent
	(i) If $\xi>\xi_*$, then
	\begin{align}
		\Re\!\bigl(2g(k)+8k^3-8\xi k\bigr)&<0,
		\qquad k\in \mathcal{C}_1\setminus\{a_1,b_1\},\label{eq:sign-super-1}\\
		\Re\!\bigl(2g(k)+8k^3-8\xi k\bigr)&>0,
		\qquad k\in \mathcal{C}_{-1}\setminus\{-a_1,-b_1\},\label{eq:sign-super-2}
	\end{align}
	where the contours $\mathcal{C}_1$ and $\mathcal{C}_{-1}$ are shown in Fig.~\ref{openinglenses1}.
	
	\smallskip
	\noindent
	(ii) If $\xi_{\mathrm{crit}}<\xi<\xi_*$, then
	\begin{align}
		\Re\!\bigl(2g(k)+8k^3-8\xi k\bigr)&>0,
		\qquad k\in \tilde{\mathcal{C}}_1\setminus\{\tilde{\alpha},b_1\},\label{eq:sign-mid-1}\\
		\Re\!\bigl(2g(k)+8k^3-8\xi k\bigr)&<0,
		\qquad k\in \tilde{\mathcal{C}}_{-1}\setminus\{-\tilde{\alpha},-b_1\},\label{eq:sign-mid-2}\\
		\Re\!\bigl(2g(k)+8k^3-8\xi k\bigr)&<0,
		\qquad k\in \tilde{\mathcal{C}}_2\setminus\{\tilde{\alpha},a_1\},\label{eq:sign-mid-3}\\
		\Re\!\bigl(2g(k)+8k^3-8\xi k\bigr)&>0,
		\qquad k\in \tilde{\mathcal{C}}_{-2}\setminus\{-\tilde{\alpha},-a_1\},\label{eq:sign-mid-4}
	\end{align}
	where the contours $\tilde{\mathcal{C}}_{\pm1}$ and $\tilde{\mathcal{C}}_{\pm2}$ are shown in Fig.~\ref{openinglenses2} and $\pm\tilde\alpha$ are the stationary points of $g(k)+\hat{\theta}(k)$.
	
	\smallskip
	\noindent
	(iii) If $\xi<\xi_{\mathrm{crit}}$, then
	\begin{align}
		\Re\!\bigl(2g(k)+8k^3-8\xi k\bigr)&>0,
		\qquad k\in \mathcal{C}_1\setminus\{a_1,b_1\},\label{eq:sign-sub-1}\\
		\Re\!\bigl(2g(k)+8k^3-8\xi k\bigr)&<0,
		\qquad k\in \mathcal{C}_{-1}\setminus\{-a_1,-b_1\},\label{eq:sign-sub-2}
	\end{align}
	where the contours $\mathcal{C}_1$ and $\mathcal{C}_{-1}$ are shown in Fig.~\ref{openinglenses1}.
\end{lemma}

\begin{proof}
	We only prove part (i), since the proofs of parts (ii) and (iii) are entirely analogous.
	
	Set
	\[
	F(k):=g(k)+\hat{\theta}(k),
	\qquad \hat{\theta}(k)=4k^3-4\xi k.
	\]
	Then
	$
	2g(k)+8k^3-8\xi k=2F(k),
	$
	so it suffices to determine the sign of $\Re F(k)$.	
	Write
	$
	k=a+ib,
	F(k)=A(a,b)+iB(a,b),
	$
	where $A=\Re F$ and $B=\Im F$. By \eqref{eq:g}, the boundary values of $F$ on
	$
	\Sigma_1\cup \Sigma_{-1}
	$
	are purely imaginary. In particular,
	$
	A(a,0)=0, a\in \Sigma_1\cup\Sigma_{-1}.
	$
	
	For $\xi>\xi_*$, Proposition~\ref{prop:q-roots} yields
	\[
	F'(k)=g'(k)+\hat{\theta}'(k)
	=12\,\frac{(k^2-k_1^2)(k^2-\alpha^2)}{R(k)},
	\]
	where
	$
	k_1<a_1, \alpha>b_1.
	$
	Now let $k\in \Sigma_1=(a_1,b_1)$. Since $R_+(k)=i|R_+(k)|$, we have
	$
	F'_+(k)
	=
	-\,12i\,\frac{(k^2-k_1^2)(k^2-\alpha^2)}{|R_+(k)|}.
	$
	Hence
	$
	B_a(a,0)=\Im F'_+(k)
	=
	-12\,\frac{(k^2-k_1^2)(k^2-\alpha^2)}{|R_+(k)|}.
	$
	Because $k\in(a_1,b_1)$, $k_1<a_1$, and $\alpha>b_1$, we have
	$
	k^2-k_1^2>0,
	k^2-\alpha^2<0,
	$
	and therefore
	$
	B_a(a,0)>0.
	$
	By the Cauchy--Riemann equations,
	$
	A_b(a,0)=-B_a(a,0)<0.
	$
	Since $A(a,0)=0$ on $\Sigma_1$, it follows that
	$
	A(a,b)<0
	$
	for $b>0$ sufficiently small, that is,
	$
	\Re F(k)<0
	$
	in the upper lens adjacent to $\Sigma_1$ and on $\mathcal{C}_1$. By the similar argument on $F_{-}(k)$, we get the  inequality holds on the whole contour
	$
	\mathcal{C}_1\setminus\{a_1,b_1\}.
	$
	This proves \eqref{eq:sign-super-1}.
	
	The proof of \eqref{eq:sign-super-2} is similar. Indeed, for $k\in \Sigma_{-1}=(-b_1,-a_1)$, the same argument gives
	$
	A_b(a,0)>0,
	$
	and hence
	$
	\Re F(k)>0
	$
	on $\mathcal{C}_{-1}\setminus\{-a_1,-b_1\}$. This proves \eqref{eq:sign-super-2}.
	
	Therefore part (i) follows, and parts (ii) and (iii) can be established in exactly the same way.
\end{proof}
In the following sections, we focus mainly on the asymptotic analysis in regions $\xi>\xi_{*}$ and $\xi_{crit}<\xi<\xi_{}*$. The analysis in region $\xi<\xi_{crit}$ is very similar to that in Region $\xi>\xi_{*}$; the main difference is that we employ a different triangular factorization of the jump matrix.
\subsection{The Supercritical Region $\xi>\xi_{*}$}\label{subsec3.1}
By using the triangular factorization of the jump matrix in \eqref{factor11} and \eqref{factor21}, we first open lenses by the transformation
\begin{align}
	Z(k)=Y(k)\begin{cases}
		\begin{pmatrix} {1}& 0\\   {i  \frac{\hat{r}(k)\hat{\rho}(k)-1}{2\hat{\rho}(k)} e^{2 t\hat{\theta}(k)} }  & {1} \end{pmatrix},\ k\in \text{lens upper of $\Sigma_1$ },\\
		\begin{pmatrix} {1}& 0\\   {-i  \frac{\hat{r}(k)\hat{\rho}(k)-1}{2\hat{\rho}(k)} e^{2 t\hat{\theta}(k)} }  & {1} \end{pmatrix},\ k\in \text{lens lower of $\Sigma_1$ },\\
		\begin{pmatrix} 1 &    -i\frac{\hat{r}(-k)\hat{\rho}(-k)-1}{2\hat{\rho}(-k)}e^{-2t\hat{\theta}(k)} \\ 0 & 1\end{pmatrix}, \ k\in \text{lens upper of $\Sigma_{-1}$ },\\
		\begin{pmatrix} 1 &    i\frac{\hat{r}(-k)\hat{\rho}(-k)-1}{2\hat{\rho}(-k)}e^{-2t\hat{\theta}(k)} \\ 0 & 1\end{pmatrix}, \ k\in \text{lens lower of $\Sigma_{-1}$ },
		\end{cases}
\end{align}
where the lenses are shown in Fig.\ref{openinglenses1}. Then $Z(k)$ has the jump conditions:
\begin{align}
	Z_{+}(k)=Z_{-}(k)\begin{cases}
		\begin{pmatrix} {1}& 0\\   {-i  \frac{\hat{r}(k)\hat{\rho}(k)-1}{2\hat{\rho}(k)} e^{2 t\hat{\theta}(k)} }  & {1} \end{pmatrix},\ k\in \mathcal{C}_1,\\
		\begin{pmatrix} 1 &    i\frac{\hat{r}(-k)\hat{\rho}(-k)-1}{2\hat{\rho}(-k)}e^{-2t\hat{\theta}(k)} \\ 0 & 1\end{pmatrix}, \ k\in\mathcal{C}_{-1},\\
		\begin{pmatrix} 0 &   {-i  \frac{2\hat{\rho}(k)}{1+\hat{r}(k)\hat{\rho}(k)}e^{-2 t\hat{\theta}(k)} }\\  -i  \le(\frac{2\hat{\rho}(k)}{1+\hat{r}(k)\hat{\rho}(k)}\ri)^{-1}e^{2 t\hat{\theta}(k)}  & 0 \end{pmatrix},\ k\in \Sigma_{1},\\
			\begin{pmatrix} 0 &  i\le(\frac{2\hat{\rho}(-k)}{1+\hat{r}(-k)\hat{\rho}(-k)}\ri)^{-1}e^{-2t\hat{\theta}(k)}\\   i\frac{2\hat{\rho}(-k)}{1+\hat{r}(-k)\hat{\rho}(-k)}e^{2t\hat{\theta}(k)} & 0\end{pmatrix},\ k\in \Sigma_{-1}.
	\end{cases}
\end{align}
Then we define a scalar function $f(k)$:
\begin{align}
f(k)=\exp\le(\frac{R(k)}{2\pi i} \le(\int_{\Sigma_{1}}\frac{\log \frac{2\hat{\rho}(\zeta)}{1+\hat{r}(\zeta)\hat{\rho}(\zeta)} }{R_{+}(\zeta)(\zeta-k)}\d \zeta + \int_{\Sigma_{-1}}\frac{-\log \frac{2\hat{\rho}(-\zeta)}{1+\hat{r}(-\zeta)\hat{\rho}(-\zeta)} }{R_{+}(\zeta)(\zeta-k)}\d \zeta +\int_{-a_1}^{a_1}\frac{\Delta}{R(\zeta)(\zeta-k)}\d\zeta \ri)\ri),
\end{align}
where
\begin{align}
	\Delta&=\big[-\int_{\Sigma_{1}}\frac{\log \frac{2\hat{\rho}(\zeta)}{1+\hat{r}(\zeta)\hat{\rho}(\zeta)} }{R_{+}(\zeta)}\d \zeta + \int_{\Sigma_{-1}}\frac{\log \frac{2\hat{\rho}(-\zeta)}{1+\hat{r}(-\zeta)\hat{\rho}(-\zeta)} }{R_{+}(\zeta)}\d \zeta\big] \big[\int_{-a_1}^{a_1}\frac{1}{R(\zeta)}\d\zeta\big]^{-1},\nonumber\\
	&=\frac{b_1}{K(m)}\int\limits_{a_1}^{b_1}\frac{\log \frac{2\hat{\rho}(\zeta)}{1+\hat{r}(\zeta)\hat{\rho}(\zeta)} }{R_{+}(\zeta)}\d \zeta\in i\R.	
\end{align}
Because $\frac{2\hat{\rho}(\zeta)}{1+\hat{r}(\zeta)\hat{\rho}(\zeta)}>0$, we get $\Delta$ is pure imaginary.
 It is easy to check that $f(k)$ satisfies
 \begin{align}
 	&f_+(k)f_-(k)=\frac{2\hat{\rho}(k)}{1+\hat{r}(k)\hat{\rho}(k)},\ k\in \Sigma_1,\\
 	&f_+(k)f_-(k)=\le(\frac{2\hat{\rho}(-k)}{1+\hat{r}(-k)\hat{\rho}(-k)}\ri)^{-1}, \ k\in \Sigma_{-1},\\
 	&\frac{f_+(k)}{f_-(k)}=e^{\Delta},\ k\in [-a_1,a_1],\\
 	&f(k)=1+\mathcal{O}(\frac{1}{k}),\ k\to\infty.
 \end{align}
Next, we introduce a new matrix-valued function
\begin{align}\label{transform}
{S}(k)=Z(k)e^{tg(k)\sigma_3}{f(k)}^{\sigma_3},\ \sigma_3=\begin{pmatrix}
	1&0\\0&-1
\end{pmatrix}.
\end{align}
By using the functions $g(k)$ and  $f(k)$, the transformation \eqref{transform} results in a new RH problem for ${S}(k)$ as follows:
\begin{RHP} Find a $1\times2$ vector-valued function ${S}(k;x,t)$ with the following properties
\begin{enumerate}
\item ${S}(k;x,t)$ is analytic for $k\in \C\backslash ([-b_1, b_1] \cup\mathcal{C}_1\cup \mathcal{C}_{-1})$.
\item For $k\in i[-b_1, b_1]$, the boundary values ${S}_{\pm}(k)$ satisfy the following jump conditions
\begin{align}
{S}_+(k)={S}_-(k) V_{{S}}(k)
\end{align}
where
\begin{equation}
\begin{array}{l}
V_{{S}}(k)=
\begin{cases}
			\begin{pmatrix} {1}& 0\\   {-i  \frac{\hat{r}(k)\hat{\rho}(k)-1}{2\hat{\rho}(k)}f^2(k) e^{2 t(\hat{\theta}(k)+g(k))} }  & {1} \end{pmatrix},\ k\in \mathcal{C}_1,\\
			\begin{pmatrix} 1 &    i\frac{\hat{r}(-k)\hat{\rho}(-k)-1}{2\hat{\rho}(-k)}f^{-2}(k)e^{-2t(\hat{\theta}(k)+g(k))} \\ 0 & 1\end{pmatrix}, \ k\in\mathcal{C}_{-1},\\
\displaystyle \begin{pmatrix} 0 &  -i\\-i &0 \end{pmatrix},  k \in  \Sigma_{1},\\
\displaystyle \begin{pmatrix} 0 &  i\\i &0 \end{pmatrix},  k \in  \Sigma_{-1},\\
\displaystyle \begin{pmatrix} e^{t\Omega+\Delta } & 0\\0& e^{-t\Omega-\Delta} \end{pmatrix},  k \in  [-a_1,a_1].
\end{cases}
\end{array}
\end{equation}
\item
${S}(k) =  \begin{pmatrix}
	1&1\end{pmatrix}  + \mathcal{O}\le(\frac{1}{k}\ri), \qquad k \rightarrow \infty.$
\end{enumerate}
\end{RHP}
By Lemma \ref{lemma3.1}, when time tends to infinity, the jumps for $S(k)$ converges to the jumps for the following outer model problem $S^{\infty}(k)$  exponentially fast  outside small neighbourhoods of $\pm ia_1$ and $\pm ib_1$.
\begin{RHP}\label{RHP4}
		Find a $1\times 2$ vector-valued function $S^{\infty}(k;x,t)$ with the following properties
\begin{enumerate}
\item {}  $S^{\infty}(k;x,t)$ is analytic in $\mathbb{C}\backslash [-b_1,b_1]$.
\item {} For $k\in [-b_1,b_1]$, the boundary values $S^{\infty}_{\pm}(k)$ satisfy the following jump relation
\begin{align}
\label{Stinfinity1}
&S^{\infty}_+(k) = S^{\infty}_-(k)V_{S^{\infty}}(k), \nonumber\\
&V_{S^{\infty}}(k)=\begin{cases}
	\begin{pmatrix}0 & -i\\ -i & 0 \end{pmatrix}, k \in\Sigma_{1},  \\
\begin{pmatrix}0 & i\\ i & 0 \end{pmatrix}, k \in\Sigma_{-1},  \\
\begin{pmatrix} e^{t\Omega+\Delta} &0 \\ 0 & e^{-t\Omega-\Delta}\end{pmatrix},   k \in [-a_1,a_1].
\end{cases}
\end{align}
\item {}$S^{\infty}(k)= \begin{pmatrix}
	1&1\end{pmatrix} +\mathcal{O}\le(\frac{1}{k}\ri), k\to\infty.$
\end{enumerate}
\end{RHP}
Next, we introduce the Jacobi elliptic function
\begin{equation}\label{eq:theta3-def}
	\vartheta_3(z;\tau)
	=
	\sum_{n\in\mathbb Z}
	e^{\,2\pi i n z+\pi i n^2\tau},
	\qquad z\in\mathbb C,\quad \Im\tau>0.
\end{equation}

The Jacobi theta function $\vartheta_3(z;\tau)$ is even in $z$ and satisfies the
quasi-periodicity relation
\begin{equation}\label{eq:theta3-periodicity}
	\vartheta_3(z+h+k\tau;\tau)
	=
	e^{-\pi i k^2\tau-2\pi i k z}\,
	\vartheta_3(z;\tau),
	\qquad h,k\in\mathbb Z.
\end{equation}

We now construct a matrix-valued outer parametrix for the model problem arising in the
super-critical regime. In contrast with the vector RH problem for $S^\infty$, the matrix
formulation is invertible and is therefore more convenient for the subsequent small-norm
analysis of the error problem.

We seek a matrix function $P^\infty(k)$, analytic for
$k\in\mathbb C\setminus(-b_1,b_1)$, such that
\begin{equation}\label{eq:Pinf-RH}
	P^\infty_+(k)=P^\infty_-(k)
	\begin{cases}
		\begin{pmatrix}
			e^{t\Omega+\Delta} & 0\\
			0 & e^{-t\Omega-\Delta}
		\end{pmatrix},
		& k\in[-a_1,a_1],\\[2ex]
		\begin{pmatrix}
			0 & -i\\
			-i & 0
		\end{pmatrix},
		& k\in \Sigma_{1},\\[2ex]
		\begin{pmatrix}
			0 & i\\
			i & 0
		\end{pmatrix},
		& k\in \Sigma_{-1},
	\end{cases}
\end{equation}
and
\begin{equation}\label{eq:Pinf-norm}
	P^\infty(k)=I+\mathcal O\!\left(\frac1k\right),
	\qquad k\to\infty.
\end{equation}

To write the solution explicitly, we first define
\begin{equation}\label{eq:tau-def}
	\tau
	=
	\frac{i}{2}\,\frac{K\!\left(\sqrt{1-m^2}\right)}{K(m)},
	\qquad
	\Omega_0
	=
	-\frac{\pi i\,b_1}{K(m)}.
\end{equation}
We also set
\begin{equation}\label{eq:w-def}
	w(k)
	=
	\int_{b_1}^{k}\frac{\Omega_0}{R(\zeta)}\,
	\frac{d\zeta}{4\pi i},
\end{equation}
and we introduce the Abelian integral
\begin{equation}\label{eq:p-def}
	p(k)
	=
	\int_{b_1}^{k}\frac{Q_1(\zeta)}{R(\zeta)}\,d\zeta,
	\qquad
	\Omega_0=-2\,p_{+}(a_1).
\end{equation}

In particular, for $k\in(-a_1,a_1)$ one has
\begin{equation}\label{eq:p-jump}
	p_{+}(k)-p_{-}(k)=-\Omega_0.
\end{equation}

We define
\begin{equation}\label{eq:gamma}
	\gamma(k)
	=
	\left(\frac{k^2-a_1^2}{k^2-b_1^2}\right)^{1/4},
\end{equation}
with the branch chosen so that $\gamma(k)\to1$ as $k\to\infty$.

We now introduce the row vector
\begin{equation}\label{eq:Sinf-supercritical}
	S^\infty(k)
	=
	\gamma(k)\,
	\frac{\vartheta_3(0;2\tau)}
	{\vartheta_3\!\left(\dfrac{t\Omega+\Delta}{2\pi i};\,2\tau\right)}
	\left(
	\frac{
		\vartheta_3\!\left(
		2w(k)+\dfrac{t\Omega+\Delta}{2\pi i}-\dfrac12;\,2\tau
		\right)}
	{
		\vartheta_3\!\left(2w(k)-\dfrac12;\,2\tau\right)
	},
	\,
	\frac{
		\vartheta_3\!\left(
		-2w(k)+\dfrac{t\Omega+\Delta}{2\pi i}-\dfrac12;\,2\tau
		\right)}
	{
		\vartheta_3\!\left(-2w(k)-\dfrac12;\,2\tau\right)
	}
	\right).
\end{equation}
Write
\[
S^\infty(k)=\bigl(S_1^\infty(k),\,S_2^\infty(k)\bigr).
\]

To construct the matrix parametrix, we define the differential operators
\begin{align}
	\nabla_{\Omega_0}S_1^\infty(k)
	&:=
	\gamma(k)\,
	\frac{\vartheta_3(0;2\tau)}
	{\vartheta_3\!\left(2w(k)-\dfrac12;\,2\tau\right)}
	\frac{\Omega_0}{2\pi i}
	\left.
	\frac{d}{dz}
	\left[
	\frac{
		\vartheta_3\!\left(
		z+2w(k)+\dfrac{t\Omega+\Delta}{2\pi i}-\dfrac12;\,2\tau
		\right)}
	{
		\vartheta_3\!\left(
		z+\dfrac{t\Omega+\Delta}{2\pi i};\,2\tau
		\right)}
	\right]
	\right|_{z=0},
	\label{eq:nabla-S1}
	\\[1ex]
	\nabla_{\Omega_0}S_2^\infty(k)
	&:=
	\gamma(k)\,
	\frac{\vartheta_3(0;2\tau)}
	{\vartheta_3\!\left(-2w(k)-\dfrac12;\,2\tau\right)}
	\frac{\Omega_0}{2\pi i}
	\left.
	\frac{d}{dz}
	\left[
	\frac{
		\vartheta_3\!\left(
		z-2w(k)+\dfrac{t\Omega+\Delta}{2\pi i}-\dfrac12;\,2\tau
		\right)}
	{
		\vartheta_3\!\left(
		z+\dfrac{t\Omega+\Delta}{2\pi i};\,2\tau
		\right)}
	\right]
	\right|_{z=0}.
	\label{eq:nabla-S2}
\end{align}

With these definitions, the outer parametrix is given by
\begin{equation}\label{eq:Pinf-explicit}
	P^\infty(k)
	=
	\frac12
	\begin{pmatrix}
		\left(1+\dfrac{p(k)}{k}\right)S_1^\infty(k)
		+\dfrac1k\,\nabla_{\Omega_0}S_1^\infty(k)
		&
		\left(1-\dfrac{p(k)}{k}\right)S_2^\infty(k)
		+\dfrac1k\,\nabla_{\Omega_0}S_2^\infty(k)
		\\[2ex]
		\left(1-\dfrac{p(k)}{k}\right)S_1^\infty(k)
		-\dfrac1k\,\nabla_{\Omega_0}S_1^\infty(k)
		&
		\left(1+\dfrac{p(k)}{k}\right)S_2^\infty(k)
		-\dfrac1k\,\nabla_{\Omega_0}S_2^\infty(k)
	\end{pmatrix}.
\end{equation}
The local parametrices $P^{\pm a_1} $ and $P^{\pm b_1}$ in neighborhoods of the branch points $\pm a_1$ and $\pm b_1$ can be constructed in terms of modified Bessel functions. The
construction is standard and the details are omitted, see for example \cite{Girotti-1}.
We now combine the outer parametrix $P^\infty(k)$ and the local parametrices near the
endpoints $\pm a_1$ and $\pm b_1$ to construct a global approximation to the RH
problem for $S(k)$.

Let the global parametrix $P(k)$ be defined by
\begin{equation}\label{eq:P-global}
	P(k)=
	\begin{cases}
		P^\infty(k), & k\in \C \setminus \bigl(B_\rho(a_1)\cup B_\rho(-a_1)\cup B_\rho(b_1)\cup B_\rho(-b_1)\bigr),\\[1ex]
		P^{a_1}(k), & k\in B_\rho(a_1),\\
		P^{-a_1}(k), & k\in B_\rho(-a_1),\\
		P^{b_1}(k), & k\in B_\rho(b_1),\\
		P^{-b_1}(k), & k\in B_\rho(-b_1),
	\end{cases}
\end{equation}
where the small disc $B_\rho(a_1)=\{k\in \C| |k-a_1|<\rho\}$ centered at $a_1$ of radius $\rho$.
We then introduce the error vector
\begin{equation}\label{eq:E-def}
	\mathcal{E}(k)=S(k)P(k)^{-1}.
\end{equation}

By construction, there exists
a constant $c>0$ such that
\begin{equation}\label{eq:E-jumps}
	\mathcal{E}_+(k)=\mathcal{E}_-(k)
	\begin{cases}
		I+\mathcal O(e^{-ct}), & \text{on the upper and lower lenses, outside the discs},\\[1ex]
		I+\mathcal O(t^{-1}), & \text{on the circles around }\pm a_1,\ \pm b_1.
	\end{cases}
\end{equation}
Moreover,
\begin{equation}\label{eq:E-norm}
	\mathcal{E}(k)=
	\begin{pmatrix}
		1 & 1
	\end{pmatrix}
	+\mathcal O\!\left(\frac{1}{k}\right),
	\qquad k\to\infty.
\end{equation}

As in the analysis of \cite{Girotti-1},  the vector $\mathcal{E}(k)$ is analytic in a neighbourhood of
$k=0$. In addition, the jump matrices for $\mathcal{E}$ satisfy the symmetry
\begin{equation}\label{eq:E-jump-symmetry}
	V_\mathcal{E}(-k)=
	\sigma_1
	V_\mathcal{E}(k)
\sigma_1.
\end{equation}
Therefore, a standard small-norm RH argument implies that the error problem admits a
unique solution and that this solution satisfies the symmetry
$
	\mathcal{E}(-k)=\mathcal{E}(k)
	\sigma_1,
$
together with the large-$k$ expansion
\begin{equation}\label{eq:E-expansion}
	\mathcal{E}(k)=
	\begin{pmatrix}
		1 & 1
	\end{pmatrix}
	+\frac{\mathcal{E}_1(x,t)}{kt}
	+\mathcal O\!\left(\frac{1}{k^2}\right),
	\qquad k\to\infty,
\end{equation}
where $\mathcal{E}_1(x,t)$ and its $x$-derivatives remain bounded as $t\to+\infty$.

We now undo the sequence of transformations. Since
\[
Y(k)=Z(k)e^{-tg(k)\sigma_3}f(k)^{-\sigma_3}
=S(k)e^{-tg(k)\sigma_3}f(k)^{-\sigma_3},\ \text{for large $k$},
\]
we have \begin{align}
	Y(k)
	&=
	\left(
	\begin{pmatrix}
		1 & 1
	\end{pmatrix}
	+\frac{\mathcal{E}_1(x,t)}{kt}
	+\mathcal O\!\left(\frac{1}{k^2}\right)
	\right)
	P(k)e^{-tg(k)\sigma_3}f(k)^{-\sigma_3}
	\nonumber\\
	&=
	\left(
	S^\infty(k)+\frac{\mathcal{E}_1(x,t)}{kt}
	+\mathcal O\!\left(\frac{1}{k^2}\right)
	\right)
	e^{-tg(k)\sigma_3}f(k)^{-\sigma_3}, \ k\to\infty.
	\label{eq:Y-large-k}
\end{align}
Hence the first component satisfies
\begin{equation}\label{eq:Y1-large-k}
	Y_1(k)=
	\left(
	S_1^\infty(k)+\frac{\mathcal{E}_1(x,t)}{kt}
	+\mathcal O\!\left(\frac{1}{k^2}\right)
	\right)e^{-tg(k)}f(k)^{-1}.
\end{equation}
Recalling the reconstruction formula
\begin{equation}\label{eq:u-recovery}
	u(x,t)=2\frac{d}{dx}\left(\lim_{k\to\infty}k\bigl(Y_1(k;x,t)-1\bigr)\right).
\end{equation}
Then, from the definition of the scalar function \(f(k)\), we obtain
\begin{equation}\label{eq:f-large-k-supercritical}
	f(k)=1+\frac{f_1(a_1,b_1)}{k}+\mathcal O\!\left(\frac{1}{k^2}\right),
\end{equation}
where \(f_1(a_1,b_1)\) is independent of \(x\).

Next, using the \(x\)-derivative formula for \(tg'(k)\), we find
\begin{equation}\label{eq:exp-tg-supercritical}
	\frac{\partial}{\partial x}e^{-tg(k)}
	=
	-\frac{1}{k}
	\left[
	\frac{a_1^2+b_1^2}{2}
	+b_1^2\left(\frac{E(m)}{K(m)}-1\right)
	\right]
	+\mathcal O\!\left(\frac{1}{k^2}\right).
\end{equation}
 We now turn to the outer model solution \(S_1^\infty(k)\). Its expansion at infinity reads
 \begin{equation}\label{eq:S1-large-k-supercritical}
 	S_1^\infty(k)
 	=
 	1+\frac{b_1}{2K(m)}\frac{1}{k}
 	\left[
 	\left(
 	\log\vartheta_3\!\left(\frac{t\Omega+\Delta}{2\pi i};\,2\tau\right)
 	\right)'
 	-
 	\frac{\vartheta_3'(0;2\tau)}{\vartheta_3(0;2\tau)}
 	\right]
 	+\mathcal O\!\left(\frac{1}{k^2}\right),
 \end{equation}
 where  $'$ denotes differentiation with respect to the argument of the theta function.
 Using the fact
 that \(\Delta\) is independent of $x,t$,
 we arrive at
 \begin{equation}\label{eq:dS1dx-supercritical}
 	\frac{\partial}{\partial x}S_1^\infty(k)
 	=
 	-\frac{1}{k}
 	\left[
 	\frac{\partial^2}{\partial x^2}
 	\log\vartheta_3\!\left(\frac{t\Omega+\Delta}{2\pi i};\,2\tau\right)
 	+\mathcal O\!\left(\frac{1}{t}\right)
 	\right]
 	+\mathcal O\!\left(\frac{1}{k^2}\right).
 \end{equation}
 Combining \eqref{eq:f-large-k-supercritical}, \eqref{eq:exp-tg-supercritical},
 and \eqref{eq:dS1dx-supercritical}, and using that \(\mathcal{E}_1(x,t)\) together with its
 \(x\)-derivatives remains bounded, we obtain
 \begin{equation}\label{eq:u-theta-supercritical}
 	u(x,t)
 	=
 	b_1^2-a_1^2
 	-2b_1^2\frac{E(m)}{K(m)}
 	-2\frac{\partial^2}{\partial x^2}
 	\log \vartheta_3\!\left(
 	\frac{b_1}{2K(m)}
 	\bigl(x-2(a_1^2+b_1^2)t+\phi_{\mathrm{sup}}\bigr)
 	;\,2\tau
 	\right)
 	+\mathcal O(t^{-1}),
 \end{equation}
 where \(K(m)\) and \(E(m)\) denote the complete elliptic integrals of the first and
 second kinds, respectively, $m=\frac{a_1}{b_1}$, and  the phase shift is
 \begin{equation}
 	\phi_{\mathrm{sup}}
 	=
 	\frac{i}{\pi}\int_{a_1}^{b_1}
 	\frac{\log \dfrac{2\hat{\rho}(\zeta)}{1+\hat{r}(\zeta)\hat{\rho}(\zeta)}}{R_+(\zeta)}\,d\zeta.
 \end{equation}
 Finally, the equivalent form \eqref{eq:uasy-super} follows from the standard
 theta-function identity
 \begin{equation}\label{eq:theta-dn-identity-supercritical}
 	\frac{1}{4K^2(m)}
 	\frac{d^2}{dz^2}\log\vartheta_3(z;2\tau)
 	=
 	-\frac{E(m)}{K(m)}
 	+\operatorname{dn}^2\!\bigl(2K(m)z+K(m)\mid m\bigr).
 \end{equation}
\subsection{The Intermediate Region $\xi_{crit}<\xi<\xi_{*}$}
 For $\xi_{crit}<\xi<\xi_{*}$, by proposition \ref{prop:q-roots}, we know the phase function $g(k)+\hat{\theta}(k)$ has the stationary points at $\pm\tilde\alpha$. Based on the  factorization \eqref{factor11}, \eqref{factor12}, \eqref{factor21} and \eqref{factor22}, we first introduce the following transformation
  \begin{align}
  	\tilde{Z}(k)=Y(k)\begin{cases}
  		\begin{pmatrix} {1}& 0\\   {i  \frac{\hat{r}(k)\hat{\rho}(k)-1}{2\hat{\rho}(k)} e^{2 t\hat{\theta}(k)} }  & {1} \end{pmatrix},\ k\in \text{lens upper of $(a_1,\tilde{\alpha})$ },\\
  		\begin{pmatrix} {1}& 0\\   {-i  \frac{\hat{r}(k)\hat{\rho}(k)-1}{2\hat{\rho}(k)} e^{2 t\hat{\theta}(k)} }  & {1} \end{pmatrix},\ k\in \text{lens lower of $(a_1,\tilde{\alpha})$ },\\
  		\begin{pmatrix} 1 &   i\frac{\hat{r}(k)\hat{\rho}(k)-1}{2\hat{r}(k)}e^{-2t\hat{\theta}(k)} \\ 0 & 1\end{pmatrix},\  k\in \text{lens upper of $(\tilde{\alpha},b_1)$ },\\
  		\begin{pmatrix} 1 &   -i\frac{\hat{r}(k)\hat{\rho}(k)-1}{2\hat{r}(k)}e^{-2t\hat{\theta}(k)} \\ 0 & 1\end{pmatrix},\ k\in \text{lens lower of $(\tilde{\alpha},b_1)$ },\\
  		\begin{pmatrix} 1 &    -i\frac{\hat{r}(-k)\hat{\rho}(-k)-1}{2\hat{\rho}(-k)}e^{-2t\hat{\theta}(k)} \\ 0 & 1\end{pmatrix}, \ k\in \text{lens upper of $(-\tilde{\alpha},-a_1)$ },\\
  		\begin{pmatrix} 1 &    i\frac{\hat{r}(-k)\hat{\rho}(-k)-1}{2\hat{\rho}(-k)}e^{-2t\hat{\theta}(k)} \\ 0 & 1\end{pmatrix}, \ k\in \text{lens lower of $(-\tilde{\alpha},-a_1)$ },\\
  		\begin{pmatrix} {1}& 0\\  {-i  \frac{\hat{r}(-k)\hat{\rho}(-k)-1}{2\hat{r}(-k)} e^{2 t\hat{\theta}(k)} } & {1} \end{pmatrix},\ k\in \text{lens upper of $(-b_1,-\tilde{\alpha})$ },\\
  		\begin{pmatrix} {1}& 0\\   {i  \frac{\hat{r}(-k)\hat{\rho}(-k)-1}{2\hat{r}(-k)} e^{2 t\hat{\theta}(k)} }  & {1} \end{pmatrix}, \ k\in \text{lens lower of $(-b_1,-\tilde{\alpha})$ }.
  	\end{cases}
  \end{align}
  Here the lens contours are shown in Figure~\ref{openinglenses2}. Then $\tilde{Z}(k)$ satisfies the jump conditions:
  \begin{align}
  	\tilde{Z}_{+}(k)=\tilde{Z}_{-}(k)\begin{cases}
  		\begin{pmatrix} {1}& 0\\   {-i  \frac{\hat{r}(k)\hat{\rho}(k)-1}{2\hat{\rho}(k)} e^{2 t\hat{\theta}(k)} }  & {1} \end{pmatrix},\ k\in \tilde{\mathcal{C}}_2,\\
  		\begin{pmatrix} 1 &   -i\frac{\hat{r}(k)\hat{\rho}(k)-1}{2\hat{r}(k)}e^{-2t\hat{\theta}(k)} \\ 0 & 1\end{pmatrix},\ k\in \tilde{\mathcal{C}}_1,\\
  		\begin{pmatrix} 1 &    i\frac{\hat{r}(-k)\hat{\rho}(-k)-1}{2\hat{\rho}(-k)}e^{-2t\hat{\theta}(k)} \\ 0 & 1\end{pmatrix}, \ k\in\tilde{\mathcal{C}}_{-2},\\
  		\begin{pmatrix} {1}& 0\\   {i  \frac{\hat{r}(-k)\hat{\rho}(-k)-1}{2\hat{r}(-k)} e^{2 t\hat{\theta}(k)} }  & {1} \end{pmatrix}, \ k\in \tilde{\mathcal{C}}_{-1},\\
  		\begin{pmatrix} 0 &   {-i  \frac{2\hat{\rho}(k)}{1+\hat{r}(k)\hat{\rho}(k)}e^{-2 t\hat{\theta}(k)} }\\  -i  \le(\frac{2\hat{\rho}(k)}{1+\hat{r}(k)\hat{\rho}(k)}\ri)^{-1}e^{2 t\hat{\theta}(k)}  & 0 \end{pmatrix},\ k\in (a_1,\tilde{\alpha}),\\
  		\begin{pmatrix} 0 &  -i\le(\frac{2\hat{r}(k)}{1+\hat{r}(k)\hat{\rho}(k)}\ri)^{-1}e^{-2t\hat{\theta}(k)}\\   -i\frac{2\hat{r}(k)}{1+\hat{r}(k)\hat{\rho}(k)}e^{2t\hat{\theta}(k)} & 0\end{pmatrix},\ k\in (\tilde{\alpha},b_1),\\
  		\begin{pmatrix} 0 &  i\le(\frac{2\hat{\rho}(-k)}{1+\hat{r}(-k)\hat{\rho}(-k)}\ri)^{-1}e^{-2t\hat{\theta}(k)}\\   i\frac{2\hat{\rho}(-k)}{1+\hat{r}(-k)\hat{\rho}(-k)}e^{2t\hat{\theta}(k)} & 0\end{pmatrix},\ k\in (-\tilde{\alpha},-a_1),\\
  		\begin{pmatrix} 0 &   {i  \frac{2\hat{r}(-k)}{1+\hat{r}(-k)\hat{\rho}(-k)}e^{-2 t\hat{\theta}(k)} }\\  i  \le(\frac{2\hat{r}(-k)}{1+\hat{r}(-k)\hat{\rho}(-k)}\ri)^{-1}e^{2t\hat{\theta}(k)}  & 0 \end{pmatrix},\ k\in (-b_1,-\tilde{\alpha}).
  	\end{cases}
  \end{align}
  For later use, we introduce the scalar function $\tilde f(k)$ by
  \begin{align}\label{def-tildef}
  	\tilde f(k)
  	=\exp\Bigg\{\frac{R(k)}{2\pi i}\Bigg[
  	&\int_{a_1}^{\tilde\alpha}
  	\frac{\log \dfrac{2\hat\rho(\zeta)}{1+\hat r(\zeta)\hat\rho(\zeta)}}
  	{R_+(\zeta)(\zeta-k)}\,d\zeta
  	-
  	\int_{\tilde\alpha}^{b_1}
  	\frac{\log \dfrac{2\hat r(\zeta)}{1+\hat r(\zeta)\hat\rho(\zeta)}}
  	{R_+(\zeta)(\zeta-k)}\,d\zeta-
  	\int_{-\tilde\alpha}^{-a_1}
  	\frac{\log \dfrac{2\hat\rho(-\zeta)}{1+\hat r(-\zeta)\hat\rho(-\zeta)}}
  	{R_+(\zeta)(\zeta-k)}\,d\zeta
  	\nonumber\\
  	&
  	+
  	\int_{-b_1}^{-\tilde\alpha}
  	\frac{\log \dfrac{2\hat r(-\zeta)}{1+\hat r(-\zeta)\hat\rho(-\zeta)}}
  	{R_+(\zeta)(\zeta-k)}\,d\zeta
  	+
  	\int_{-a_1}^{a_1}
  	\frac{\tilde\Delta}{R(\zeta)(\zeta-k)}\,d\zeta
  	\Bigg]\Bigg\}.
  \end{align}
  Here the constant $\tilde\Delta$ is chosen so that $\tilde f(k)$ is normalized at infinity, namely
  \begin{align}\label{def-tDelta}
  	\tilde\Delta
  	=
  	\Bigg[
  	&\int_{\tilde\alpha}^{b_1}
  	\frac{\log \dfrac{2\hat r(\zeta)}{1+\hat r(\zeta)\hat\rho(\zeta)}}{R_+(\zeta)}\,d\zeta
  	-
  	\int_{a_1}^{\tilde\alpha}
  	\frac{\log \dfrac{2\hat\rho(\zeta)}{1+\hat r(\zeta)\hat\rho(\zeta)}}{R_+(\zeta)}\,d\zeta
  	\nonumber\\
  	&\quad
  	+
  	\int_{-\tilde\alpha}^{-a_1}
  	\frac{\log \dfrac{2\hat\rho(-\zeta)}{1+\hat r(-\zeta)\hat\rho(-\zeta)}}{R_+(\zeta)}\,d\zeta
  	-
  	\int_{-b_1}^{-\tilde\alpha}
  	\frac{\log \dfrac{2\hat r(-\zeta)}{1+\hat r(-\zeta)\hat\rho(-\zeta)}}{R_+(\zeta)}\,d\zeta
  	\Bigg]
  	\Bigg[
  	\int_{-a_1}^{a_1}\frac{d\zeta}{R(\zeta)}
  	\Bigg]^{-1}.
  \end{align}
  Using the symmetry of the integrands, this can be rewritten as
  \begin{align}\label{def-tDelta-simplified}
  	\tilde\Delta
  	=
  	\frac{b_1}{K(m)}
  	\left(
  	\int_{a_1}^{\tilde\alpha}
  	\frac{\log \dfrac{2\hat\rho(\zeta)}{1+\hat r(\zeta)\hat\rho(\zeta)}}{R_+(\zeta)}\,d\zeta
  	-
  	\int_{\tilde\alpha}^{b_1}
  	\frac{\log \dfrac{2\hat r(\zeta)}{1+\hat r(\zeta)\hat\rho(\zeta)}}{R_+(\zeta)}\,d\zeta
  	\right).
  \end{align}

  A direct computation shows that $\tilde f(k)$ satisfies the following jump relations:
  \begin{align}
  	\tilde f_+(k)\tilde f_-(k)
  	&=
  	\frac{2\hat\rho(k)}{1+\hat r(k)\hat\rho(k)},
  	&& k\in (a_1,\tilde\alpha),\\
  	\tilde f_+(k)\tilde f_-(k)
  	&=
  	\left(
  	\frac{2\hat r(k)}{1+\hat r(k)\hat\rho(k)}
  	\right)^{-1},
  	&& k\in (\tilde\alpha,b_1),\\
  	\tilde f_+(k)\tilde f_-(k)
  	&=
  	\left(
  	\frac{2\hat\rho(-k)}{1+\hat r(-k)\hat\rho(-k)}
  	\right)^{-1},
  	&& k\in (-\tilde\alpha,-a_1),\\
  	\tilde f_+(k)\tilde f_-(k)
  	&=
  	\frac{2\hat r(-k)}{1+\hat r(-k)\hat\rho(-k)},
  	&& k\in (-b_1,-\tilde\alpha),\\
  	\frac{\tilde f_+(k)}{\tilde f_-(k)}
  	&=
  	e^{\tilde\Delta},
  	&& k\in [-a_1,a_1],\\
  	\tilde f(k)
  	&=
  	1+\mathcal{O}\!\left(\frac1k\right),
  	&& k\to\infty.
  \end{align}
  \begin{figure}
  	\begin{tikzpicture}[
  		scale=1.15,
  		line cap=round,
  		line join=round,
  		>=Stealth,
  		curvearrow/.style={
  			thick,
  			postaction={decorate},
  			decoration={markings, mark=at position 0.58 with {\arrow{>}}}
  		}
  		]
  		
  		\def\xa{2.0}   
  		\def\xp{4.2}   
  		\def\xb{6.6}   
  		
  		\def\hs{0.82}  
  		\def\hb{1.15}  
  		
  		\coordinate (Lm) at (-\xb,0);   
  		\coordinate (Lp) at (-\xp,0);   
  		\coordinate (La) at (-\xa,0);   
  		
  		\coordinate (Ra) at (\xa,0);    
  		\coordinate (Rp) at (\xp,0);    
  		\coordinate (Rb) at (\xb,0);    
  		
  		\draw[thick] (Lm) -- (Lp) -- (La);
  		\draw[thick] (Ra) -- (Rp) -- (Rb);
  		
  		\draw[curvearrow] (Lm) .. controls (-5.95,\hb) and (-4.95,\hb) .. (Lp);
  		\draw[curvearrow] (Lm) .. controls (-5.95,-\hb) and (-4.95,-\hb) .. (Lp);
  		
  		\draw[curvearrow] (Lp) .. controls (-3.45,\hs) and (-2.65,\hs) .. (La);
  		\draw[curvearrow] (Lp) .. controls (-3.45,-\hs) and (-2.65,-\hs) .. (La);
  		
  		\draw[curvearrow] (Ra) .. controls (2.65,\hs) and (3.45,\hs) .. (Rp);
  		\draw[curvearrow] (Ra) .. controls (2.65,-\hs) and (3.45,-\hs) .. (Rp);
  		
  		\draw[curvearrow] (Rp) .. controls (4.95,\hb) and (5.95,\hb) .. (Rb);
  		\draw[curvearrow] (Rp) .. controls (4.95,-\hb) and (5.95,-\hb) .. (Rb);
  		
  		\node[above] at (-5.35,1.02) {$\tilde{\mathcal{C}}_{-1}$};
  		\node[above] at (-3.05,0.76) {$\tilde{\mathcal{C}}_{-2}$};
  		
  		\node at (-5.35,-0.28) {$\Sigma_{-1}$};
  		\node at (-3.05,-0.28) {$\Sigma_{-1}$};
  		
  		\node[left]  at (Lm) {$-b_1$};
  		\node[right] at (La) {$-a_1$};
  		\node[below] at (-4.2,-0.2) {$-\tilde{\alpha}$};
  		
  		\node[above] at (3.05,0.76) {$\tilde{\mathcal{C}}_{2}$};
  		\node[above] at (5.35,1.02) {$\tilde{\mathcal{C}}_{1}$};
  		
  		\node at (3.05,-0.28) {$\Sigma_{1}$};
  		\node at (5.35,-0.28) {$\Sigma_{1}$};
  		
  		\node[left]  at (Ra) {$a_1$};
  		\node[right] at (Rb) {$b_1$};
  		\node[below] at (4.2,-0.2) {$\tilde{\alpha}$};
  		
  	\end{tikzpicture}
  	\caption{Opening lenses  for $\xi\in (\xi_{crit},\xi_{*})$.}
  	\label{openinglenses2}
  \end{figure}
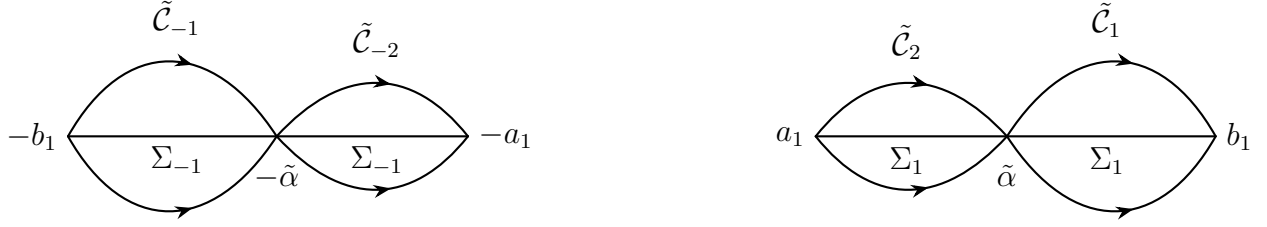
 Next, we introduce the transformation
 \begin{equation}\label{def-S-tilde}
 	\tilde S(k)=\tilde Z(k)e^{t g(k)\sigma_3}\tilde f(k)^{\sigma_3}.
 \end{equation}
 Then $\tilde S(k)$ satisfies the jump relation
 \begin{equation}\label{jump-S-tilde}
 	\tilde S_+(k)=\tilde S_-(k)\,\tilde V_S(k),
 \end{equation}
 where
 \begin{equation}\label{jump-matrix-S-tilde}
 	\tilde V_S(k)=
 	\begin{cases}
 		\begin{pmatrix}
 			1 & 0\\[1mm]
 			-\displaystyle i\,\frac{\hat r(k)\hat\rho(k)-1}{2\hat\rho(k)}
 			\,\tilde f(k)^2 e^{2t(\hat\theta(k)+g(k))} & 1
 		\end{pmatrix},
 		& k\in \tilde{\mathcal C}_2,
 		\\[5mm]
 		\begin{pmatrix}
 			1 &
 			-\displaystyle i\,\frac{\hat r(k)\hat\rho(k)-1}{2\hat r(k)}
 			\,\tilde f(k)^{-2} e^{-2t(\hat\theta(k)+g(k))}\\[1mm]
 			0 & 1
 		\end{pmatrix},
 		& k\in \tilde{\mathcal C}_1,
 		\\[5mm]
 		\begin{pmatrix}
 			1 &
 			\displaystyle i\,\frac{\hat r(-k)\hat\rho(-k)-1}{2\hat\rho(-k)}
 			\,\tilde f(k)^{-2} e^{-2t(\hat\theta(k)+g(k))}\\[1mm]
 			0 & 1
 		\end{pmatrix},
 		& k\in \tilde{\mathcal C}_{-2},
 		\\[5mm]
 		\begin{pmatrix}
 			1 & 0\\[1mm]
 			\displaystyle i\,\frac{\hat r(-k)\hat\rho(-k)-1}{2\hat r(-k)}
 			\,\tilde f(k)^2 e^{2t(\hat\theta(k)+g(k))} & 1
 		\end{pmatrix},
 		& k\in \tilde{\mathcal C}_{-1},
 		\\[5mm]
 		\begin{pmatrix}
 			0 & -i\\
 			-i & 0
 		\end{pmatrix},
 		& k\in \Sigma_1,
 		\\[5mm]
 		\begin{pmatrix}
 			0 & i\\
 			i & 0
 		\end{pmatrix},
 		& k\in \Sigma_{-1},
 		\\[5mm]
 		\begin{pmatrix}
 			e^{t\Omega+\tilde\Delta} & 0\\
 			0 & e^{-t\Omega-\tilde\Delta}
 		\end{pmatrix},
 		& k\in [-a_1,a_1].
 	\end{cases}
 \end{equation}

 By Lemma~\ref{lemma3.1}, the off-diagonal entries of the jump matrix on
 $
 \tilde{\mathcal C}_{\pm1}\cup \tilde{\mathcal C}_{\pm2}
$
 decay exponentially as $t\to+\infty$  outside of small neighborhoods of the endpoints 
$\pm a_1$, $\pm b_1$
 and of the stationary phase point $\pm \tilde{\alpha}$.
The analysis of local parametrix  near the points $\pm a_1$, $\pm b_1$
is very similar to the analysis done in \cite{Girotti-1}.
Since $\pm\tilde\alpha$ are the stationary points of $g(k)+\hat{\theta}(k)$, the local parametrices in neighborhoods of $\pm\tilde\alpha$ can be constructed in terms of parabolic cylinder functions.
\subsection{The local model near $\tilde{\alpha}$}
Let $D_{\epsilon}(\tilde{\alpha})$ denote the small disc centered at $\tilde{\alpha}$ of radius $\epsilon$. Let 
$$X_{1}\cup X_5=\tilde{\mathcal{C}}_1\cap D_{\epsilon}(\tilde{\alpha}),\ \ X_{2}\cup X_4=\tilde{\mathcal{C}}_2\cap D_{\epsilon}(\tilde{\alpha}),\ \ X_3\cup X_6=\Sigma_{1}\cap D_{\epsilon}(\tilde{\alpha}).$$
These curves divide the disk $D_{\epsilon}(\tilde{\alpha})$ into regions $\{R_j\}_1^6$, see Fig.~\ref{fig:disk-regions}. Define $g_{\tilde{\alpha}}(k)$ for $k$ near $\tilde{\alpha}$ by
\begin{align}
	g_{\tilde{\alpha}}(k)=\begin{cases}
		g(k)+\hat{\theta}(k)-g_{+}(\alpha)-\hat{\theta}(\tilde{\alpha}), \ k\in R_1\cup R_2\cup R_3,\\
		-(g(k)+\hat{\theta}(k)-g_{-}(\alpha)-\hat{\theta}(\tilde{\alpha})),\ k\in R_4\cup R_5 \cup R_6,
	\end{cases}
\end{align}
and define the function $M^{(1)}(k)$ for $k$ near $\tilde{\alpha}$ by
\begin{align}
	M^{(1)}(k)=\tilde{S}(k)\tilde{f}^{-\sigma_3}(k)e^{it(g_{+}(\tilde{\alpha})+\hat{\theta}(\tilde{\alpha}))\sigma_3},\ \ k\in D_{\epsilon}(\tilde{\alpha})\backslash (\bigcup\limits_{j=1}^{6}X_j).
\end{align}
Across the curves $\bigcup\limits_{j=1}^{6}X_j$, $M^{(1)}(k)$ satisfies the jump condition $M^{(1)}_{+}(k)=M^{(1)}_{-}(k)V^{(1)}$ where 
\begin{align}
	V^{(1)}=\begin{cases}
		\begin{pmatrix}
			1&-i\frac{\hat r(k)\hat\rho(k)-1}{2\hat r(k)}e^{-2tg_{\tilde{\alpha}}(k)}\\0&1
		\end{pmatrix},\ k\in X_1,\\
		\begin{pmatrix}
			1& -i\frac{\hat r(k)\hat\rho(k)-1}{2\hat r(k)}e^{2tg_{\tilde{\alpha}}(k)}\\ 0&1
		\end{pmatrix},\ k\in X_5,\\
		\begin{pmatrix}
			1&0\\ -i\frac{\hat r(k)\hat\rho(k)-1}{2\hat \rho(k)}e^{2tg_{\tilde{\alpha}}(k)}&1
		\end{pmatrix},\ k\in X_2,\\
		\begin{pmatrix}
			1&0\\ -i\frac{\hat r(k)\hat\rho(k)-1}{2\hat \rho(k)}e^{-2tg_{\tilde{\alpha}}(k)}&1
		\end{pmatrix},\ k\in X_4,\\
		\begin{pmatrix}
			0&-i \frac{2\hat{\rho}(k)}{1+\hat r(k)\hat\rho(k)}\\
			-i (\frac{2\hat{\rho}(k)}{1+\hat r(k)\hat\rho(k)})^{-1} &0
		\end{pmatrix},\ k\in X_3,\\
		\begin{pmatrix}
			0&-i (\frac{2\hat{r}(k)}{1+\hat r(k)\hat\rho(k)})^{-1}\\
			-i \frac{2\hat{r}(k)}{1+\hat r(k)\hat\rho(k)}
			 &0
		\end{pmatrix},\ k\in X_6,
	\end{cases}
\end{align}
where all contours are oriented right as in Fig.~\ref{fig:disk-regions}.
Then we introduce the complex-valued function $\delta(k)$ by
\begin{align}
	\delta(k)=\exp(\frac{1}{2\pi i}(\int_{a_1}^{\tilde{\alpha}}\ln(\frac{2\hat{\rho}(s)}{1+\hat r(s)\hat\rho(s)})\frac{1}{s-k}\d s-\int_{\tilde{\alpha}}^{b_1}\ln(\frac{2\hat{r}(s)}{1+\hat r(s)\hat\rho(s)})\frac{1}{s-k}\d s)).
\end{align}
Since $\ln (\frac{2\hat{\rho}(\tilde{\alpha})}{1+\hat r(\tilde{\alpha})\hat\rho(\tilde{\alpha})})>0$ and $\ln (\frac{2\hat{r}(\tilde{\alpha})}{1+\hat r(\tilde{\alpha})\hat\rho(\tilde{\alpha})})>0$, by Lemma C.1 of \cite{Boutet4}, we know that $\delta(k)$ is bounded and analytic for $k\in D_{\epsilon}(\tilde{\alpha})\backslash (a_1,b_1)$.
A direct computation shows that 
\begin{align}
	\frac{\delta_{+}(k)}{\delta_{-}(k)}=\begin{cases}
		\frac{2\hat{\rho}(k)}{1+\hat r(k)\hat\rho(k)},\ k\in (a_1,\tilde{\alpha}),\\
		(\frac{2\hat{r}(k)}{1+\hat r(k)\hat\rho(k)})^{-1},\ k\in (\tilde{\alpha},b_1).
	\end{cases}
\end{align}
Then we define the piecewise constant functions $A(k)$ and $B(k)$ by
\begin{align}
	A(k)=\begin{cases}
		\delta^{\sigma_3}(k),\ k\in R_1\cup R_2 \cup R_3,\\
		\delta^{-\sigma_3}(k),\ k\in R_4\cup R_5 \cup R_6,
	\end{cases},\\
	B(k)=\begin{cases}
		\begin{pmatrix}
			0&i\\i&0
		\end{pmatrix},\ k\in R_1\cup R_2 \cup R_3,\\
		I,\ k\in R_4\cup R_5 \cup R_6,
	\end{cases}.
\end{align}
\begin{figure}[htbp]
	\centering
	\begin{tikzpicture}[scale=1.1,>=stealth,line cap=round,line join=round]
		
		\def\r{3}
		
		\coordinate (O) at (0,0);
		
		\coordinate (X1p) at (50:\r);
		\coordinate (X2p) at (130:\r);
		\coordinate (X3p) at (180:\r);
		\coordinate (X4p) at (230:\r);
		\coordinate (X5p) at (320:\r);
		\coordinate (X6p) at (0:\r);
		
		\draw[thick] (O) circle (\r);
		
		\draw[thick] (O) -- (X1p);
		\draw[thick] (O) -- (X2p);
		\draw[thick] (O) -- (X3p);
		\draw[thick] (O) -- (X4p);
		\draw[thick] (O) -- (X5p);
		\draw[thick] (O) -- (X6p);
		
		\draw[->,thick] ($(O)!0.35!(X1p)$) -- ($(O)!0.68!(X1p)$);
		\draw[->,thick] ($(O)!0.35!(X5p)$) -- ($(O)!0.68!(X5p)$);
		\draw[->,thick] ($(O)!0.35!(X6p)$) -- ($(O)!0.68!(X6p)$);
		
		\draw[->,thick] ($(O)!0.68!(X2p)$) -- ($(O)!0.35!(X2p)$);
		\draw[->,thick] ($(O)!0.68!(X3p)$) -- ($(O)!0.35!(X3p)$);
		\draw[->,thick] ($(O)!0.68!(X4p)$) -- ($(O)!0.35!(X4p)$);
		
		\node[above right] at ($(X1p)+(0.25,0.25)$) {$X_1$};
		\node[above left]  at ($(X2p)+(-0.25,0.25)$) {$X_2$};
		\node[left]        at ($(X3p)+(-0.55,0)$) {$X_3$};
		\node[below left]  at ($(X4p)+(-0.25,-0.25)$) {$X_4$};
		\node[below right] at ($(X5p)+(0.25,-0.25)$) {$X_5$};
		\node[right]       at ($(X6p)+(0.55,0)$) {$X_6$};
		
		\node at (25:1.9)  {$R_1$};
		\node at (90:1.8)  {$R_2$};
		\node at (155:1.7) {$R_3$};
		\node at (205:1.7) {$R_4$};
		\node at (270:1.8) {$R_5$};
		\node at (335:1.8) {$R_6$};
		
	\end{tikzpicture}
	\caption{The contour in the disk $D_{\epsilon}(\tilde{\alpha})$ and the sets $\{R_j\}_{1}^{6}$}
	\label{fig:disk-regions}
\end{figure}
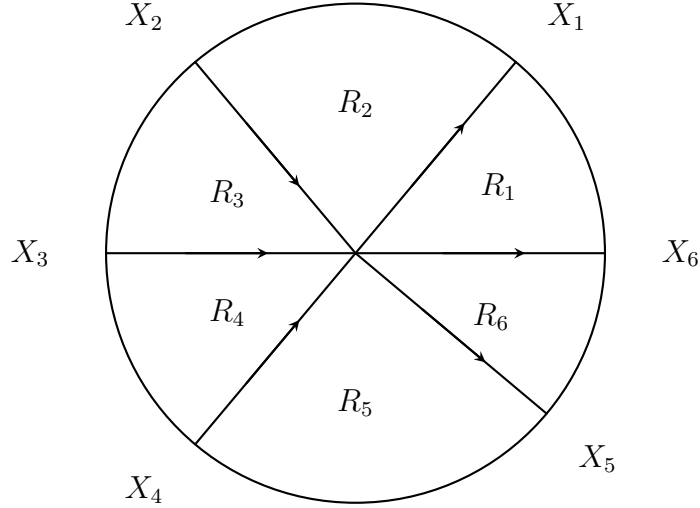
Define a new vector-valued function $M^{(2)}(k)$ by
\begin{align}
	M^{(2)}(k)=M^{(1)}(k)A(k)B(k),
\end{align}
then the jumps on the contours $X_3$ and $X_6$ are eliminated, we obtain $M^{(2)}(k)$ has the jump condition
\begin{align}
	M_{+}^{(2)}(k)=M^{(2)}_{-}(k)\begin{cases}
	\begin{pmatrix}
		1&0\\
		-i\frac{\hat r(k)\hat\rho(k)-1}{2\hat r(k)}\delta^{-2}(k)e^{-2tg_{\tilde{\alpha}}(k)}&1
	\end{pmatrix},\ k\in X_1,\\
	\begin{pmatrix}
		1& -i\frac{\hat r(k)\hat\rho(k)-1}{2\hat r(k)}\delta^{2}(k)e^{2tg_{\tilde{\alpha}}(k)}\\ 0&1
	\end{pmatrix},\ k\in X_5,\\
	\begin{pmatrix}
		1& -i\frac{\hat r(k)\hat\rho(k)-1}{2\hat \rho(k)}\delta^{2}(k)e^{2tg_{\tilde{\alpha}}(k)}\\
		0&1
	\end{pmatrix},\ k\in X_2,\\
	\begin{pmatrix}
		1&0\\ -i\frac{\hat r(k)\hat\rho(k)-1}{2\hat \rho(k)}\delta^{-2}(k) e^{-2tg_{\tilde{\alpha}}(k)}&1
	\end{pmatrix},\ k\in X_4,\\
	I,\ k\in X_3\cup X_6.
	\end{cases}
\end{align}
In order to relate $M^{(2)}$ to the model problem $m^{X}$ of Appendix A in \cite{Boutet4}, we make a local change of variables for $k$ near $\tilde{\alpha}$ and introduce the new variable $\zeta=\zeta(k)$ by
$$\zeta=\sqrt{itg_{\tilde{\alpha}+}(k)},$$
where the branch of the square root is fixed as follows: Since $g_{\tilde{\alpha}+}(k)$ has a double zero at $k=\tilde{\alpha}$, we can write $\zeta=\sqrt{t}(k-\tilde{\alpha})\psi_{\tilde{\alpha}}(k)$, where the function $\psi_{\tilde{\alpha}}(k)$ is analytic for $k\in D_{\epsilon}(\tilde{\alpha})$. From the definition \eqref{gprime} of $g(k)$, we get $\frac{\d^2 g_{\tilde{\alpha}+}(k)}{\d k^2}(\tilde{\alpha})\in -i \R_{+}$, hence we can fix the branch by requiring that $\psi_{\tilde{\alpha}}(\tilde{\alpha})>0$.

The map $k\to\zeta$ is a biholomorphism from the disk $D_{\epsilon}(\tilde{\alpha})$ onto a neighborhood of the origin, which maps $X_3$ and $X_6$ into $R_-$ and $R_+$, respectively. By deforming the contours $X_1, X_2, X_4$, and  $X_5$, we can assume that they are mapped into the straight rays of the complex $\zeta$-plane for which $\arg \zeta$ equals $\frac{\pi}{4}$, $\frac{3\pi}{4}$, $-\frac{3\pi}{4}$, and $-\frac{\pi}{4}$, respectively.

We next consider the behavior of $\delta(k)$ as $k$ approaches $\tilde{\alpha}$. Let $\ln _{a_1}(k-\tilde{\alpha})$ and $\ln _{b_1}(k-\tilde{\alpha})$ denote the function $\ln (k-\tilde{\alpha})$ with branch cut along $(a_1,\tilde{\alpha})$ and $(\tilde{\alpha},b_1)$, respectively. We also define two functions $L_{a_1}$ and $L_{b_1}$ by
\begin{align*}
	 L_{a_1}(s,k)=\ln (k-s),\ s\in (a_1,\tilde{\alpha}),\ \ k\in D_{\epsilon}(\tilde{\alpha})\backslash (a_1,\tilde{\alpha}),\\
	 L_{b_1}(s,k)=\ln (k-s),\ s\in (\tilde{\alpha},b_1),\ \ k\in D_{\epsilon}(\tilde{\alpha})\backslash (\tilde{\alpha},b_1),
\end{align*}
where the branches are fixed by requiring that:
\begin{itemize}
	\item $L_{a_1}(s,k)$ is a continuous function of $s\in (a_1,\tilde{\alpha})$ for each $k\in D_{\epsilon}(\tilde{\alpha})\backslash (a_1,\tilde{\alpha})$.
	\item $L_{b_1}(s,k)$ is a continuous function of $s\in (\tilde{\alpha},b_1)$ for each $k\in D_{\epsilon}(\tilde{\alpha})\backslash (\tilde{\alpha},b_1)$.
	\item For $s=\tilde{\alpha}$, we have $L_{a_1}(\tilde{\alpha},k)=\ln _{a_1}(k-\tilde{\alpha})$ and $L_{b_1}(\tilde{\alpha},k)=\ln_{b_1}(k-\tilde{\alpha})$.
\end{itemize}
Integrating by parts, we obtain
\begin{align*}
	\int_{a_1}^{\tilde{\alpha}}\ln(\frac{2\hat{\rho}(s)}{1+\hat r(s)\hat\rho(s)})\frac{1}{s-k}\d s=&\ln_{a_1}(k-\tilde{\alpha})\ln(\frac{2\hat{\rho}(\tilde{\alpha})}{1+\hat r(\tilde{\alpha})\hat\rho(\tilde{\alpha})})-L_{a_1}(a_1,k) \ln(\frac{2\hat{\rho}(a_1)}{1+\hat r(a_1)\hat\rho(a_1)})\\&-\int_{a_1}^{\tilde{\alpha}} L_{a_1}(s,k)\d \ln(\frac{2\hat{\rho}(s)}{1+\hat r(s)\hat\rho(s)}),
\end{align*}
and 
\begin{align*}
	\int_{\tilde{\alpha}}^{b_1}\ln(\frac{2\hat{r}(s)}{1+\hat r(s)\hat\rho(s)})\frac{1}{s-k}\d s=&L_{b_1}(b_1,k) \ln(\frac{2\hat{r}(b_1)}{1+\hat r(b_1)\hat\rho(b_1)})-\ln_{b_1}(k-\tilde{\alpha})\ln(\frac{2\hat{r}(\tilde{\alpha})}{1+\hat r(\tilde{\alpha})\hat\rho(\tilde{\alpha})})\\&-\int_{\tilde{\alpha}}^{b_1} L_{b_1}(s,k)\d \ln(\frac{2\hat{r}(s)}{1+\hat r(s)\hat\rho(s)}).
\end{align*}
Hence we can write 
\begin{align}
	\delta(k)=\exp (\frac{1}{2\pi i}(\ln_{a_1}(k-\tilde{\alpha})\ln(\frac{2\hat{\rho}(\tilde{\alpha})}{1+\hat r(\tilde{\alpha})\hat\rho(\tilde{\alpha})})+\ln_{b_1}(k-\tilde{\alpha})\ln(\frac{2\hat{r}(\tilde{\alpha})}{1+\hat r(\tilde{\alpha})\hat\rho(\tilde{\alpha})})+ \chi (k))),
\end{align}
where the function $\chi(k)$ is defined by
\begin{align}
	\chi(k)=&\frac{1}{2\pi i}(-L_{a_1}(a_1,k) \ln(\frac{2\hat{\rho}(a_1)}{1+\hat r(a_1)\hat\rho(a_1)})-\int_{a_1}^{\tilde{\alpha}} L_{a_1}(s,k)\d \ln(\frac{2\hat{\rho}(s)}{1+\hat r(s)\hat\rho(s)})\nonumber\\&-L_{b_1}(b_1,k) \ln(\frac{2\hat{r}(b_1)}{1+\hat r(b_1)\hat\rho(b_1)})+\int_{\tilde{\alpha}}^{b_1} L_{b_1}(s,k)\d \ln(\frac{2\hat{r}(s)}{1+\hat r(s)\hat\rho(s)})).
\end{align}
Let $\ln_{a}(k)$ denote the logarithm of $k$ with branch cut along $\arg k=a$, i.e., $\ln_{a} k=\ln |k|+i \arg k$ with $\arg k \in (a,a+2\pi]$. Then $\ln_{-\pi}(k)=\ln k$. Since the map $k\to \zeta$ takes $X_3$ into $R_-$ and $X_6$ into $R_+$, we have 
\begin{align*}
	\ln_{a_1}(k-\tilde{\alpha})=\ln_{a_1}\frac{\zeta}{\sqrt{t}\psi_{\tilde{\alpha}}(k)}=-\frac{\ln t}{2}+\ln \zeta-\ln \psi_{\tilde{\alpha}}(k),\ k\in D_{\epsilon}(\tilde{\alpha})\backslash (a_1,\tilde{\alpha}),\\
	\ln_{b_1}(k-\tilde{\alpha})=\ln_{b_1}\frac{\zeta}{\sqrt{t}\psi_{\tilde{\alpha}}(k)}=-\frac{\ln t}{2}+\ln_{0} \zeta-\ln \psi_{\tilde{\alpha}}(k),\ k\in D_{\epsilon}(\tilde{\alpha})\backslash (\tilde{\alpha},b_1).
\end{align*}
Define the function $$p(\zeta)=\exp (\frac{1}{2\pi i}(\ln(\frac{2\hat{\rho}(\tilde{\alpha})}{1+\hat r(\tilde{\alpha})\hat\rho(\tilde{\alpha})}) \ln \zeta + \ln(\frac{2\hat{r}(\tilde{\alpha})}{1+\hat r(\tilde{\alpha})\hat\rho(\tilde{\alpha})}) \ln_0\zeta)),\ \zeta\in \C\backslash \R.$$
Moreover, define the functions $\delta_0(t)$ and $\delta_1(k)$ by
\begin{align*}
	\delta_0(t)=\exp (\frac{1}{2\pi i}((\ln(\frac{2\hat{\rho}(\tilde{\alpha})}{1+\hat r(\tilde{\alpha})\hat\rho(\tilde{\alpha})})+\ln(\frac{2\hat{r}(\tilde{\alpha})}{1+\hat r(\tilde{\alpha})\hat\rho(\tilde{\alpha})}))(-\frac{1}{2}\ln t-\ln \psi_{\tilde{\alpha}}(\tilde{\alpha})))+\chi (\tilde{\alpha})),\ t>0,\\
	\delta_1(k)=\exp ( \frac{1}{2\pi i}(-(\ln(\frac{2\hat{\rho}(\tilde{\alpha})}{1+\hat r(\tilde{\alpha})\hat\rho(\tilde{\alpha})})+\ln(\frac{2\hat{r}(\tilde{\alpha})}{1+\hat r(\tilde{\alpha})\hat\rho(\tilde{\alpha})}))(\ln \psi_{\tilde{\alpha}}(k)-\ln \psi_{\tilde{\alpha}}(\tilde{\alpha})) )+ \chi(k)-\chi (\tilde{\alpha}) ).
\end{align*}
Hence $\delta(k)$ has the form
\begin{align}
	\delta(k)=p(\zeta(k))\delta_0(t)\delta_1(k).
\end{align} 
Define $M^{(3)}(k)$ by 
\begin{align*}
	M^{(3)}(\zeta(k))=M^{(2)}(k)\delta_0(t)^{\sigma_3}(\frac{\hat{r}(\tilde{\alpha})}{\hat{\rho}(\tilde{\alpha})})^{\frac{1}{4}\sigma_3},
\end{align*} 
then $M^{(3)}(k)$ satisfies the jump condition $M^{(3)}_+(k)=M_-^{(3)}(k)V^{(3)}$ where
\begin{align*}
	V^{(3)}=\begin{cases}
		\begin{pmatrix}
			1&0\\
			i\sqrt{\frac{\hat{r}(\tilde{\alpha})}{\hat{\rho}(\tilde{\alpha})}}\frac{\hat r(\zeta)\hat\rho(\zeta)-1}{2\hat r(\zeta)}\delta_1^{-2}(\zeta) p^{-2}(\zeta)e^{2i\zeta^2}&1
		\end{pmatrix},\ \zeta\in \arg \zeta=\frac{\pi}{4},\\
		\begin{pmatrix}
			1& i\sqrt{\frac{\hat{\rho}(\tilde{\alpha})}{\hat{r}(\tilde{\alpha})}}\frac{\hat r(\zeta)\hat\rho(\zeta)-1}{2\hat r(\zeta)}\delta_1^{2}(\zeta) p^{2}(\zeta)e^{-2i\zeta^2}\\ 0&1
		\end{pmatrix},\ \zeta\in \arg \zeta =-\frac{\pi}{4},\\
		\begin{pmatrix}
			1& -i\sqrt{\frac{\hat{\rho}(\tilde{\alpha})}{\hat{r}(\tilde{\alpha})}}\frac{\hat r(\zeta)\hat\rho(\zeta)-1}{2\hat \rho(\zeta)}\delta_1^{2}(\zeta) p^{2}(\zeta)e^{-2i\zeta^2}\\
			0&1
		\end{pmatrix},\ \zeta\in \arg \zeta=\frac{3\pi}{4},\\
		\begin{pmatrix}
			1&0\\ -i\sqrt{\frac{\hat{r}(\tilde{\alpha})}{\hat{\rho}(\tilde{\alpha})}}\frac{\hat r(\zeta)\hat\rho(\zeta)-1}{2\hat \rho(\zeta)}\delta_1^{-2}(\zeta) p^{-2}(\zeta) e^{2i\zeta^2}&1
		\end{pmatrix},\ \zeta\in \arg \zeta=-\frac{3\pi}{4},
	\end{cases}
\end{align*}
where all contours are oriented toward the origin. Define the function $\tilde{p}$ by
$$ \tilde{p}(\zeta)=\exp (i \nu \ln_{-\frac{\pi}{2}}(\zeta))\\
,$$
where the constant $\nu=\nu(\tilde{\alpha})$ is defined by
\begin{align}\label{nu}
	\nu=-\frac{1}{2\pi}\ln(\frac{2\hat{\rho}(\tilde{\alpha})}{1+\hat r(\tilde{\alpha})\hat\rho(\tilde{\alpha})})-\frac{1}{2\pi}\ln(\frac{2\hat{r}(\tilde{\alpha})}{1+\hat r(\tilde{\alpha})\hat\rho(\tilde{\alpha})})>0,
\end{align}
where we have used the fact $\frac{4\hat{\rho}(\tilde{\alpha})\hat{r}(\tilde{\alpha})}{(1+\hat r(\tilde{\alpha})\hat\rho(\tilde{\alpha}))^2}<1$.
Since 
\begin{align*}
	p(\zeta)=\tilde{p}(\zeta)\begin{cases}
		1,\ \arg \zeta \in (0,\pi),\\
		\frac{1+\hat r(\tilde{\alpha})\hat\rho(\tilde{\alpha})}{2\hat{\rho}(\tilde{\alpha})}, \ \arg \zeta \in (-\pi,-\frac{\pi}{2}),\\
		\frac{2\hat{r}(\tilde{\alpha})}{1+\hat r(\tilde{\alpha})\hat\rho(\tilde{\alpha})},\ \arg \zeta \in (-\frac{\pi}{2},0),
	\end{cases}
\end{align*}
we can write 
\begin{align*}
	V^{(3)}=\begin{cases}
		\begin{pmatrix}
			1&0\\
			i\sqrt{\frac{\hat{r}(\tilde{\alpha})}{\hat{\rho}(\tilde{\alpha})}}\frac{\hat r(\zeta)\hat\rho(\zeta)-1}{2\hat r(\zeta)}\delta_1^{-2}(\zeta) \tilde{p}^{-2}(\zeta)e^{2i\zeta^2}&1
		\end{pmatrix},\ \zeta\in \arg \zeta=\frac{\pi}{4},\\
		\begin{pmatrix}
			1& i\sqrt{\frac{\hat{\rho}(\tilde{\alpha})}{\hat{r}(\tilde{\alpha})}}\frac{\hat r(\zeta)\hat\rho(\zeta)-1}{2\hat r(\zeta)} (\frac{2\hat{r}(\tilde{\alpha})}{1+\hat r(\tilde{\alpha})\hat\rho(\tilde{\alpha})})^2\delta_1^{2}(\zeta) \tilde{p}^{2}(\zeta)e^{-2i\zeta^2}\\ 0&1
		\end{pmatrix},\ \zeta\in \arg \zeta =-\frac{\pi}{4},\\
		\begin{pmatrix}
			1& -i\sqrt{\frac{\hat{\rho}(\tilde{\alpha})}{\hat{r}(\tilde{\alpha})}}\frac{\hat r(\zeta)\hat\rho(\zeta)-1}{2\hat \rho(\zeta)}\delta_1^{2}(\zeta) \tilde{p}^{2}(\zeta)e^{-2i\zeta^2}\\
			0&1
		\end{pmatrix},\ \zeta\in \arg \zeta=\frac{3\pi}{4},\\
		\begin{pmatrix}
			1&0\\ -i\sqrt{\frac{\hat{r}(\tilde{\alpha})}{\hat{\rho}(\tilde{\alpha})}}\frac{\hat r(\zeta)\hat\rho(\zeta)-1}{2\hat \rho(\zeta)}\delta_1^{-2}(\zeta) (\frac{2\hat{\rho}(\tilde{\alpha})}{1+\hat r(\tilde{\alpha})\hat\rho(\tilde{\alpha})})^2 \tilde{p}^{-2}(\zeta) e^{2i\zeta^2}&1
		\end{pmatrix},\ \zeta\in \arg \zeta=-\frac{3\pi}{4}.
	\end{cases}
\end{align*}
Define $q(\tilde{\alpha})=i\frac{\hat{r}(\tilde{\alpha})\hat\rho(\tilde{\alpha})-1}{\sqrt{\hat{r}(\tilde{\alpha})\hat\rho(\tilde{\alpha})}}$. For a fixed $\zeta$, $i\sqrt{\frac{\hat{r}(\tilde{\alpha})}{\hat{\rho}(\tilde{\alpha})}}\frac{\hat r(\zeta)\hat\rho(\zeta)-1}{2\hat r(\zeta)}\to q(\tilde{\alpha})$ and $\delta_1^{2}(\zeta)\to 1$ as $t\to\infty$. This suggests that $V^{(3)}$ tends to the jump matrix $v^{X}$ defined in Appendix A of \cite{Boutet4} for large $t$. In conclusion, the local parametrix around the endpoint $k=\tilde{\alpha}$ is
\begin{align}
	P^{\tilde{\alpha}}(k)=Y_{\tilde{\alpha}}(k)m^{X}(q,\zeta(k))\delta_0(t)^{-\sigma_3}(\frac{\hat{r}(\tilde{\alpha})}{\hat{\rho}(\tilde{\alpha})})^{-\frac{1}{4}\sigma_3}B^{-1}(k)A^{-1}(k)e^{-it(g(\tilde{\alpha})+\hat{\theta}(\tilde{\alpha}))\sigma_3}\tilde{f}^{\sigma_3}(k),
\end{align}
where $m^{X}$ is the solution of the RH problem (A.2) of \cite{Boutet4} and $Y_{\tilde{\alpha}}(k)$ is a function which is analytic for $k\in D_{\epsilon}(\tilde{\alpha})$, which is determined by imposing that $P^{\tilde{\alpha}}(k)(P^{\infty}(k))^{-1}=I+\mathcal{O}(t^{-\frac{1}{2}}) $ on $\partial D_{\epsilon}(\tilde{\alpha})$ as $t\to\infty$. Therefore, we choose
\begin{align}
	Y_{\tilde{\alpha}}(k)=P^{\infty}(k)\tilde{f}^{-\sigma_3}(k)e^{it(g(\tilde{\alpha})+\hat{\theta}(\tilde{\alpha}))\sigma_3}A(k)B(k)\delta_0(t)^{\sigma_3}(\frac{\hat{r}(\tilde{\alpha})}{\hat{\rho}(\tilde{\alpha})})^{\frac{1}{4}\sigma_3}.
\end{align}
The functions $P^{\infty}(k)\tilde{f}^{-\sigma_3}(k)$ and $(A(k)B(k))^{-1}$ have the same jump across $X_3\cup X_6$, thus $Y_{\tilde{\alpha}}(k)$ is analytic in $D_{\epsilon}(\tilde{\alpha})$.

 Combining the local analysis with a standard small-norm argument (see example \cite{AE}), we arrive at the following asymptotic representation of the function $\tilde{S}(k)$
\begin{align}
	\tilde{S}(k)=\le(\begin{pmatrix} 1&1\end{pmatrix}+\frac{\tilde{\mathcal{E}}_1(x,t)}{k\sqrt{t}}+\mathcal{O}(k^{-2})\ri)P^{\infty}(k),\ k\to\infty,
\end{align} 
where $P^{\infty}(k)$ is given by \eqref{eq:Pinf-explicit} and  $\tilde{\mathcal{E}}_1(x,t)$ and its $x$-derivatives remain bounded as $t\to+\infty$.
Follow similar argument as section \ref{subsec3.1}, we obtain the large-time asymptotic formula \eqref{eq:uasy-mid}.

\section*{Acknowledgments}
This work is supported by the National Natural Science Foundation of China (Grant Nos. 12471234, 12271490, 12471240), the Excellent Youth Science Fund Project of Henan Province (Grant No. 242300421158) and Science Foundation of Henan Academy of Sciences (Grant No. 20252319002).

\section*{Conflict of interest}
The authors declare no conflicts of interest.

\section*{Data availability statement}
Data sharing is not applicable to this article as no new data were created
or analyzed during the current study.

	\end{document}